\documentclass[letterpaper,11pt,final]{article}

\newif\ifdraft
\drafttrue % or \draftfalse
\usepackage{typearea}
\typearea{15}
\usepackage[compact]{titlesec}

\usepackage{setspace}
\usepackage{latexsym,graphicx}
\usepackage{amsthm,amsmath,amssymb,enumerate,mathrsfs}
\usepackage{thmtools,thm-restate}
\usepackage{empheq}
\usepackage{xspace}
\usepackage{bm}
\usepackage{ifpdf}
\usepackage[usenames,dvipsnames]{xcolor}
\usepackage{algorithm}
\usepackage[noend]{algpseudocode}
\usepackage{nicefrac}
\usepackage{wrapfig}

\usepackage{fullpage}

\allowdisplaybreaks

\definecolor{Darkblue}{rgb}{0,0,0.4}
\definecolor{Brown}{cmyk}{0,0.81,1.,0.60}
\definecolor{Purple}{cmyk}{0.45,0.86,0,0}

\ifpdf
 
\fi
\usepackage[breaklinks]{hyperref}
\hypersetup{colorlinks=true,%pdfborder={1 1 1 [3]},%
            citebordercolor={.6 .6 .6},linkbordercolor={.6 .6 .6},%
citecolor=blue,urlcolor=black,linkcolor=NavyBlue,pagecolor=black}


\newtheorem{theorem}{Theorem}[section]

\newtheorem{lemma}[theorem]{Lemma}

\newtheorem{observation}[theorem]{Observation}
\newtheorem{claim}[theorem]{Claim}

\newtheorem{corollary}[theorem]{Corollary}

\newtheorem{invariant}[theorem]{Invariant}

\numberwithin{algorithm}{section}

\newcommand{\junk}[1]{}
\newcommand{\ignore}[1]{}

\newcommand{\R}[0]{{\ensuremath{\mathbb{R}}}}

\newcommand{\E}[0]{{\ensuremath{\mathbb{E}}}}

\newcommand{\poly}{\operatorname{poly}}

\newcommand{\sse}{\subseteq}

\newcommand{\calC}{{\mathscr{C}}}

\newcommand{\e}{\varepsilon}
\newcommand{\eps}{\varepsilon}
\newcommand{\tsty}{\textstyle}

\ifdraft
\newcounter{note}[section]
\renewcommand{\thenote}{\thesection.\arabic{note}}
\newcommand{\agnote}[1]{\refstepcounter{note}$\ll${\bf Anupam~\thenote:}
  {\sf \color{red} #1}$\gg$\marginpar{\tiny\bf AG~\thenote}}
\newcommand{\nbnote}[1]{\refstepcounter{note}$\ll${\bf Niv~\thenote:}
  {\sf \color{green} #1}$\gg$\marginpar{\tiny\bf NB~\thenote}}
\definecolor{purple}{rgb}{0.7,0,0.8}
\newcommand{\mmnote}[1]{\refstepcounter{note}$\ll${\bf Marco~\thenote:}
  {\sf \color{purple} #1}$\gg$\marginpar{\tiny\bf MM~\thenote}}
\newcommand{\snnote}[1]{\refstepcounter{note}$\ll${\bf Seffi~\thenote:}
  {\sf \color{red} #1}$\gg$\marginpar{\tiny\bf SN~\thenote}}
\else
\newcommand{\agnote}[1]{}
\newcommand{\nbnote}[1]{}
\newcommand{\mmnote}[1]{}
\newcommand{\snnote}[1]{}
\fi

\newcommand{\alert}[1]{{\color{red}#1}}

\newcommand{\qedsymb}{\hfill{\rule{2mm}{2mm}}}

\newcommand{\initOneLiners}{%
    \setlength{\itemsep}{0pt}
    \setlength{\parsep }{0pt}
    \setlength{\topsep }{0pt}
}
\newenvironment{OneLiners}[1][\ensuremath{\bullet}]
    {\begin{list}
        {#1}
        {\initOneLiners}}
    {\end{list}}

\newcommand{\squishlist}{
 \begin{list}{$\bullet$}
  { \setlength{\itemsep}{0pt}
     \setlength{\parsep}{3pt}
     \setlength{\topsep}{3pt}
     \setlength{\partopsep}{0pt}
     \setlength{\leftmargin}{1.5em}
     \setlength{\labelwidth}{1em}
     \setlength{\labelsep}{0.5em} } }

\newcommand{\squishend}{
  \end{list}  }

\newcommand{\gr}{\nabla}
\newcommand{\ip}[1]{\langle #1 \rangle}

\DeclarePairedDelimiterX{\infdivx}[2]{(}{)}{%
  #1\;\delimsize\|\;#2%
}
\newcommand{\Ddiv}{D\infdivx}
\newcommand{\DXdiv}[1]{D_{#1}\infdivx}
\newcommand{\Pdiv}{\Phi\infdivx}
\DeclarePairedDelimiter{\norm}{\lVert}{\rVert}
\renewcommand{\a}{\mathbf{x}}
\renewcommand{\b}{\mathbf{y}}
\newcommand{\x}{\mathbf{x}}
\newcommand{\y}{\mathbf{y}}
\newcommand{\z}{\mathbf{z}}
\newcommand{\wlone}[1]{\norm{#1}_{\ell_1(w)}}
\newcommand{\Tlone}[1]{\norm{#1}_{\ell_t(T)}}

\newcommand{\ones}{\bm{1}}
\newcommand{\half}{\nicefrac12}

\begin{document}

\title{{\bf $k$-Servers with a Smile: \\Online Algorithms via Projections}}

\author{Niv Buchbinder\thanks{Dept. of Statistics and Operations Research, Tel Aviv University, Israel.} \and Anupam Gupta\thanks{Computer Science Department, Carnegie Mellon University,
Pittsburgh, USA. Supported in part by NSF awards CCF-1536002, CCF-1540541,
and CCF-1617790, and the Indo-US Joint Center for Algorithms Under Uncertainty.} \and Marco Molinaro\thanks{PUC-Rio, Rio de Janeiro, Brazil. Supported in part by CNPq grants Universal \#431480/2016-8 and Bolsa de Produtividade em Pesquisa \#310516/2017-0, and FAPERJ grant Jovem Cientista do Nosso Estado.} \and Joseph (Seffi) Naor\thanks{Computer Science Department, Technion, Israel.}}

\maketitle

\begin{abstract}
  We consider the $k$-server problem on trees and HSTs. We give 
  an algorithm based on 
  Bregman projections. This algorithm has a competitive ratios that
  match some of the recent results given by Bubeck et al. (STOC 2018),
  whose algorithm was based on
  mirror-descent-based continuous dynamics prescribed via a differential
  inclusion.
\end{abstract}

\section{Introduction}
\label{sec:introduction}

The $k$-server problem is one of the cornerstones of online algorithms and competitive analysis.
It captures many other classic online problems (like paging) that maintain
``feasible'' stages while satisfying a sequence of requests arriving
online. Given a metric space $(X,d)$ on
$n = |X|$ points, the input is a sequence of requests
$r_1, r_2, \ldots, r_T, \ldots$, where each request $r_t$ is a point in
the metric space. The algorithm maintains a set $A_t \sse X$ of $k$
points in the metric, which gives the locations of the $k$
\emph{servers}. We require that $r_t \in A_t$ for all $t$, i.e., at each
time-step there is a server at the requested point. The cost of the
algorithm is the sum of the (earthmover) distances between the
consecutive states of the algorithm; i.e., the total distance
traveled by the $k$ servers while occupying locations $A_t$ at time $t$.

The problem has a long rich history; we list some relevant events and
refer to~\cite{BBMN11,BCLLM17} for more references.  Manasse et
al.~\cite{MMS} introduced it and conjectured a (deterministic)
$k$-competitive algorithm for all metrics; there is a deterministic
lower bound of $k$ even for the uniform metric (which is equivalent to
the \emph{paging} problem). This conjecture technically still remains
open, though the $(2k-1)$-competitive algorithm of Koutsoupias and
Papadimitriou~\cite{KP} settled it in spirit. The focus then shifted to
the randomized $k$-server conjecture: \emph{can we get an
  $O(\log k)$-competitive randomized algorithm for general metrics?}
Such a result would be tight, since there is a lower bound of
$\Omega (\log k)$, again coming from the paging problem.  The first
non-trivial improvement over the deterministic case was a
$\poly\log(k,n)$-competitive randomized algorithm which is due to Bansal
et al.~\cite{BBMN11}. Very recently, Bubeck et al.~\cite{BCLLM17}
achieved a breakthrough, introducing several new ideas to give an
algorithm that is $O(D \log k)$-competitive on depth-$D$ trees, and
$O(\log^2 k)$-competitive on HSTs. Moreover, they gave a dynamic
tree-embedding result to show an
$O(\log^3 k \log \Delta)$-competitiveness result for general
metrics. Subsequently, Lee~\cite{Lee18} employed further new ideas to
remove the dependence on $\Delta$ and achieve an
$O(\log^6 k)$-competitive randomized algorithm for general metrics. This
is the first $\poly\log k$-competitive algorithm for general metrics.

The~\cite{BCLLM17} paper defined a differential inclusion, whose (unique)
solution gives a fractional solution to an LP relaxation for the
problem. Our first result is a different algorithm (albeit directly
inspired by theirs) with the same aymptotic guarantees. 

We first give the result for trees with small hop-diameter:
\begin{theorem}[Low-Depth Trees]
  \label{thm:main}
  There is a deterministic algorithm that outputs a fractional
  solution for the $k$-server problem when the metric $(X,d)$ is a tree
  metric, with competitive ratio $O(D \log k)$, where $D$ is the
  hop-diameter of the tree.
\end{theorem}
Recall that the hop-diameter of a tree is the maximum number of edges on
any simple path in the tree. A useful sub-class of trees are
$\tau$-HSTs; these have a designated root, and consecutive edge-lengths decrease
by a factor of $\tau > 1$ along any root-leaf path. (We
suppress the $\tau$ and just refer to HSTs when the precise value of
$\tau$ is not important.) It is easy to transform any $\tau$-HST
into one that has depth $O(\log n)$, while changing distances by a factor
of at most $\frac{2\tau}{\tau - 1}$. Moreover, on such trees it is
possible to randomly round fractional solutions to integer ones using
ideas from~\cite{BBMN11,BCLLM17}. This implies the following result:
\begin{corollary}
  \label{cor:HSTs}
  There is a randomized algorithm that is
  $O(\min\{D,\log n\} \log k)$-competitive for the $k$-server problem
  when the metric space $(X,d)$ is induced by an HST. Again, $D$ is the
  hop-diameter of the tree.
\end{corollary}

Finally, we can improve this guarantee to get an $O(\log^2 k)$ guarantee:
\begin{theorem}[HSTs]
  \label{thm:logsq}
  There is a randomized algorithm that is
  $O(\log^2 k)$-competitive for the $k$-server problem when the metric
  space $(X,d)$ is induced by a $\tau$-HST for $\tau \leq 1/10$.
\end{theorem}

\medskip\textbf{Bicriteria Problems.} Our algorithm naturally extends to the \emph{$(h,k)$-server} problem, where the algorithm has $k$
servers, but its cost is compared to the cost of the best solution with
only $h \leq k$ servers. For the weighted star metric, this
problem admits $\frac{k}{k-h+1}$-competitive
deterministic~\cite{ST85,Y94} and $O(\log \frac{k}{k-h+1})$-competitive
randomized algorithms~\cite{BBN-simple}; these guarantees
approach $1$ as $k/h \to \infty$. For more general tree metrics, such
strong guarantees are not possible.  Bansal et al.~\cite{BEJK} showed a
lower bound of $2.41$ on the competitiveness of deterministic algorithms
for the $(h,k)$-server problem on depth-$2$ HSTs, even when
$h \ll k$. They also gave a deterministic algorithm for depth-$D$ trees
with competitive ratio $D (1 - 1/((1+\e)^{1/D} - 1))^{D+1}$ where
$k/h = (1+\e)$. E.g., for $\e \in (0,1]$, this factor is about
$D (2D/\e)^{D+1}$; contrast this with the deterministic
$(1/\e)$-competitiveness for paging.

A small change to the algorithm and analysis from Theorems~\ref{thm:main}
and~\ref{thm:logsq} gives the following:

\begin{theorem}
  \label{thm:main-hk}
  There is a deterministic algorithm that outputs a fractional
  solution for the $(h,k)$-server problem when the metric $(X,d)$ is a
  tree metric, with competitive ratio $O(D \log (1/\e))$, where
  $k/h = 1+\e$. For the case of HSTs, we get a bound of
  $O(\min\{D, \log k\} \log (1/\e))$, and can round it to get a randomized
  algorithm with the same asymptotic competitive ratio.
\end{theorem}

We point out that the algorithms of~\cite{BCLLM17} also extend to the
$(h,k)$-server setting and give a competitive factor of $O(\log k \log
(1/\e))$~\cite{BubPC}.

\medskip\textbf{Techniques.}
The algorithm is easy to state. We use the elegant linear programming
relaxation of the $k$-server problem given by~\cite{BCLLM17}, which
defines a feasible polytope $P$ amenable to online computation. 
Let $\x^t$ denote the ``anti-server" solution at time $t$\footnote{The indicator $x^t_u=0$ means that there is a server at node $u$ and $x^t_u=1$ indicates otherwise.}.
The points that serve the request at
node $r_t$ at time $t$ are those in some subspace $P \cap \{x: x_{r_t}= 0\}$.  Now,
given the previous solution $\x^{t-1}$, getting a solution $\x^t$ for
time $t$ is easy: we project $\a^{t-1}$ onto this subspace
$P \cap \{x: x_{r_t}= 0\}$. The projection is not a Euclidean
projection, but is with respect to a ``natural'' distance function in
the context of trees---the distance corresponding to the (negative)
\emph{multiscale entropy} function $D$. Formally,
\[ \a^t \gets \arg\min_{x \in P \cap \{x: x_{r_t}=0\}}
  \Ddiv{x}{\a^{t-1}}. \] This amounts to solving a convex
program\footnote{We are glossing over an important detail---we will
  need to project on the affine slice $P \cap \{x: x_{r_t}= \delta\}$
  for some positive small $\delta$ to ensure the gradients are
  Lipschitz. See
  \S\ref{sec:server} for the full story.}. Note that this
projection-based algorithm makes a sequence of discrete jumps, one for each time step, and hence
differs from the
mirror-descent approach of~\cite{BCLLM17} which takes
infinitesimal steps and keeps using Bregman projections to get back
into $P$, until feasibility is achieved.

The analysis also draws significantly on~\cite{BCLLM17}, with some
differences because of the discrete steps. The proof of competitiveness is via a potential-function
argument. This potential is (more or less) the distance between the
optimal solution and that of the algorithm (measured according to the
distance function used for the projection):
\[ \Phi_t := \Ddiv{OPT_t}{ALG_t}. \] However, since we project at each
time using a distance that is a Bregman divergence, we use the
``reverse-Pythagorean'' property of such distance functions to relate
the distance between $\a^{t-1}$ and $\a^t$ to the drop in potential, and
thereby show
\[ \Ddiv{ALG_t}{ALG_{t-1}} + (\Phi_t - \Phi_{t-1}) \leq \alpha\cdot
  \Delta OPT. \] (To obtain this inequality we also need to relate the change in $OPT$ to the change in potential, which we do by choosing divergences $\Ddiv{\cdot}{\cdot}$ that are ``Lipschitz'' in the first argument.) The technical work is then to relate the actual
movement cost $\| ALG_t - ALG_{t-1} \|_1$ to this Bregman distance
$\Ddiv{ALG_t}{ALG_{t-1}}$. These proofs are short for set cover and
weighted paging (which we give here for completeness and intuition), and
longer for $k$-server. We also show that our algorithm
maintains many of the side invariants that hold for the
continuous process of~\cite{BCLLM17}.

We emphasize that Bregman projections are not new in the context of
online algorithms: they have been
explicitly used by, e.g.,~\cite{BCN14, BCLLM17}, and also implicitly
underlie many online primal-dual algorithms for packing and covering
problems, even though the algorithms may not be explicitly described 
in this language. Moreover, the projection method is not a
panacea: there are problems for which this approach does not seem to
give us the desired fine-grained control over the solutions; for
other problems like MTS we need to incorporate service costs. Yet, we
hope this perspective will be useful in other contexts.

Finally, the authors of~\cite{BCLLM17} inform us that their algorithm
can also be discretized, by repeatedly taking $\e = \e(n,k)$-sized
steps, and then (Bregman) projecting back onto the polytope of feasible
points~\cite{JamesPC}. Indeed, since the proof
of~\cite[Theorem~5.6]{BCLLM17} proceeds by finding a discrete
sequence of feasible points and then taking limits, this proof can be 
used to show that such a discretization process works.

\paragraph{Roadmap.}
As a warm-up we give a rephrasing of the primal-dual algorithm for the
(unweighted) set cover problem in terms of projections. This provides
the basic ingredients of the analysis: how the KKT conditions
allow us to prove useful properties of the projected points, which in turn are
used in the analysis based on the reverse-Pythagorean property of
Bregman projections. We then build on these ideas to give algorithms for
the $(h,k)$-weighted paging problem in \S\ref{sec:paging}, the
$(h,k)$-server problem on trees in \S\ref{sec:server}, and for HSTs in
\S\ref{sec:logsq}.  We defer the paging example to an appendix to get to
$k$-server earlier, but the non-expert reader may want to read the
paging example first to gain some more familiarity with the ideas.

We emphasize that the length of some of our
proofs %for set cover and paging
comes from spelling out all the details. E.g., readers familiar with
basics of convex optimization can easily compress the proofs for set
cover and paging to a page each; other proofs  also can be 
considerably shortened.

\subsection{Related Work}

The fact that many primal-dual algorithms could be viewed as mirror
descent was observed by Buchbinder, Chen, and Naor~\cite{BCN14}, who
used this viewpoint to give approximation algorithms for \emph{set cover
  with service costs}, which is a simultaneous extension of set cover
and metrical task systems (MTS) on a weighted star.
Many authors have used convex programs
to analyze and solve online problems, e.g., Devanur and Jain~\cite{DJ}, Anand et al.~\cite{AGK12},
Gupta et al.~\cite{GKP11}, Devanur and Huang~\cite{DH14}, Kim and
Huang~\cite{HuangK15}, Im et al.~\cite{IKM}, and others have developed
primal-dual and dual-fitting techniques using convex programs. However, the ideas
and techniques used there are different from the ones in this paper.

%%% Local Variables:
%%% mode: latex
%%% TeX-master: "main"
%%% End:

\subsection{Notation and Preliminaries}
\label{sec:prelims}

We merely give some definitions that we need in this paper; for details
about convexity and convex optimization, see Rockafellar~\cite{Rock},
Hiriart-Urruty and Lemarechal~\cite{HUL}, or Boyd and
Vanderberghe~\cite{BV}.

Given a convex function $h: \R^n \to \R$, the \emph{Bregman divergence}
associated with $h$ is given by
\[ \DXdiv{h}{p}{q} := h(p) - h(q) - \ip{ \gr h(q), p-q }. \] In words,
this is the amount by which the linear approximation of convex function
$h$ at point $q$ underestimates the function value at $p$.  We get an
\emph{under}estimate because $h$ is convex. Hence the divergence is
non-negative for any $p,q$, and is zero when $p=q$. (If the function $h$
is strictly convex, then the converse also holds, i.e.,
$\DXdiv{h}{p}{q} = 0 \implies p=q$.) The \emph{Bregman projection} of a
point $q$ onto a convex body $P$ is simply
$p = \arg\min_{x \in P} \DXdiv{h}{x}{q}$.

One commonly used Bregman divergence between non-negative vectors
is the \emph{(unnormalized) Kullback-Liebler divergence}:
\[ \Ddiv{p}{q} = \sum_i \Big( p_i \log \frac{p_i}{q_i} - p_i + q_i \Big)
\] which arises from the \emph{negative entropy function}
$h(p) = \sum_i p_i \log p_i$. If we consider $p$ and $q$ in the
probability simplex (or indeed, if $p$ and $q$ have the same $\ell_1$
norm), then the linear terms fall away and $\Ddiv{p}{q}$ becomes the
well-known \emph{normalized KL-divergence} (or relative entropy
function) $\sum_i p_i \log \frac{p_i}{q_i}$.  Now, projecting a point $x$
onto the probability simplex using KL-divergence, simply scales up
each coordinate by the same factor---this is the multiplicative-weight
update rule! We will use variants of KL-divergence extensively.

We need the \emph{reverse-Pythagorean property} of Bregman
divergences. Given a convex body $K$, with point $y \in K$, and the
Bregman projection $x' := \min_{z \in K} \Ddiv{z}{x}$ for some point
$x$,
\[ \Ddiv{y}{x} \geq \Ddiv{y}{x'} + \Ddiv{x'}{x}. \] The name comes from
the illuminating (visual) proof of this fact for the squared Euclidean
distance Bregman divergence
$\Ddiv{p}{q} := \frac12 \norm{p-q}_2^2$.

Using the inequality $1 + x \leq e^x$ with $x = \log (b/a)$
gives what we call the ``poor-man's Pinsker'' inequality: for all
$a, b \geq 0$,
\begin{gather}
  a - b \leq a \log \frac{a}{b}. \tag{PMP} \label{eq:pmp}
\end{gather}

%%% Local Variables:
%%% mode: latex
%%% TeX-master: "main"
%%% End:

\section{The Unweighted Set Cover Problem}
\label{sec:unweighted-set-cover}

As a warm-up, we solve and analyze (unweighted) set
cover in the projection perspective. This is to show the main steps:
the derived algorithm is essentially that from Alon et al.~\cite{AAABN},
and the projection viewpoint was already noted in, e.g.,~\cite{BCN14}.
The steps in subsequent sections are similar to those here, only more
involved.

We are given $n$ sets, and ``set
covering'' constraints $\ip{a_t, x} \geq 1$ arrive online, where each
request vector $a_{t} \in \{0,1\}^n$. This defines the set covering
polytope:
\[ P_t := \{ x \geq 0 \mid \ip{a_s, x} \geq 1 \; \forall s\leq t \}. \]
The goal is to maintain a fractional solution $\a^t \in P_t$ that is
monotone (i.e., $\a^t \geq \a^{t-1}$) and is an approximately good
solution to the linear program:
\[ \min_{x \in P_t} \ip{\mathbf{1}, x}. \] We compare ourselves to the
optimal integer solution $\b^t \in \{0,1\}^n$.  We start with the
initial solution $\a^0_i = \delta$, where $\delta = \frac{1}{n}$; hence
we start off with the fractional solution buying one set in total. As
long as the request sequence contains at least one request, this extra
(fractional) set does not affect the competitive ratio except by at most a factor of $2$.

\subsection{The Projection Algorithm}

Project the old point $\a^{t-1}$ onto the new body $P_t$
using the Bregman divergence:
\[ \Ddiv{x}{x'} := \sum_i \bigg(x_i \log \frac{x_i}{x'_i} - x_i +
  x'_i\bigg) \] I.e., set
$\a^t := \arg\min_{x \in P_t} \Ddiv{x}{\a^{t-1}}$. Since all entries
$a_{ti} \geq 0$ (i.e., the polytope is a covering polytope), we claim
this
projection is equivalent to just projecting onto the convex set
$Q_t := \R^n \cap \{ \ip{a_t, x} \geq 1\}$. To see this, define
$z := \arg\min_{z \in Q_t} \Ddiv{z}{\a^{t-1}}$. We show
\emph{monotonicity} below: that $z_i \geq \a^{t-1}_i$ for all $i$, and
hence $z \in P_{t-1}$. Since $P_{t-1} \cap Q_t = P_t$, this shows an
equivalence between the two projections, and hence $\a^t = z$.

To show monotonicity, consider the point $z$: it is the solution to
the convex program
  \begin{align*}
    \min \quad \sum_{i} \Big(z_i \log &\frac{z_i}{\a^{t-1}_i} -
                 z_i + \a^{t-1}_i\Big) \\
    \sum_i a_{ti} z_i &\geq 1
  \end{align*}
The KKT optimality conditions imply:
\begin{gather}
  \log \frac{z_i}{\a^{t-1}_i} = \lambda_t \, a_{ti}, \label{eq:kkt-sc}
\end{gather}
where the Lagrange multiplier $\lambda_t \geq 0$ satisfies the
complementary slackness condition %$\gamma_i(z_i - 1) = 0$ and
$\lambda_t(\sum_i a_{ti} z_i - 1) = 0$. (Since the constraints are
affine and only the objective function is convex, strong duality holds
as long as the problem is feasible. Hence we do not have problems with
duality gaps, and so can assume the existence of optimal duals/KKT
certificates/multipliers for the problems in this paper.)

Rewriting~(\ref{eq:kkt-sc}),
$z_i = \a^{t-1}_i e^{\lambda_t a_{ti}} \geq \a^{t-1}_i$ because the
exponent is non-negative. This proves monotonicity, and the preceding
argument then gives $\a^t = z$. Moreover, defining
$P^\delta_t = P_t \cap [\delta,1]^n$, we get that $\a^t \in P^\delta_t$
for all $t$.

\subsection{Analysis}
\label{sec:analysis}

We use a potential function to measure the ``distance'' from the optimal
solution to the algorithm's solution. Define the function
\[ \Pdiv{x}{x'} = \sum_i x_i \log \frac{x_i}{x'_i}. \]
Observe that $\Pdiv{x}{x'} = \Ddiv{x}{x'} + \ip{ \ones, x-x'}$.
Define the potential after serving the $t^{th}$ request to be:
\[ \Pdiv{\b^t}{\a^t} := \sum_i \b^t_i \log \frac{\b^t_i}{\a^t_i} =
  \sum_{i: \b^t_i = 1} \log \frac{1}{\a^t_i}. \] Here $\b^t$ is the
optimal integer solution for $P_t$.  Since the fractional
solution $\a^t \in P^\delta_t$, each term of this potential is
non-negative, and at most $\log 1/\delta$. (As an aside, we could set
$\Pdiv{\cdot}{\cdot} = \Ddiv{\cdot}{\cdot}$ with tiny changes, but we
find the current view cleaner.)

\paragraph{When OPT moves:}
We can ensure that OPT only changes entries from $0$ to
$1$. In this case the increase in potential is at most $\log 1/\delta$, so
we get
\[ \Pdiv{\b^t}{\a^{t-1}} - \Pdiv{\b^{t-1}}{\a^{t-1}} \leq (\log 1/\delta)
  \cdot \Delta OPT. \tag{$\star$}\]

\paragraph{When ALG moves:}
Since we used a Bregman divergence to project a point $\a^{t-1}$ down to
$\a^t \in P_t$ (and also because the optimal point $\b^t \in P_t$), we
can use the reverse-Pythagorean property of these projections to claim
\[ \Ddiv{\b^t}{\a^{t-1}} \geq \Ddiv{\b^t}{\a^t} + \Ddiv{\a^t}{\a^{t-1}}. \]
Substituting the definition of $\Ddiv{\cdot}{\cdot}$ and canceling linear terms from both sides, we get
\[ \Pdiv{\b^t}{\a^{t-1}} \geq \Pdiv{\b^t}{\a^t} + \Pdiv{\a^t}{\a^{t-1}}, \]
or equivalently,
\[ \underbrace{\Pdiv{\a^t}{\a^{t-1}}}_{\text{``shadow'' cost}} +
  \underbrace{\big(\Pdiv{\b^t}{\a^t} - \Pdiv{\b^t}{\a^{t-1}} \big)}_{\text{change
      in potential}} \leq 0. \tag{$\star\star$}\]
% Hence it suffices to show that the actual cost is bounded by this shadow
% cost, i.e.,
% \begin{gather}
%   \norm{\a^t - \a^{t-1}}_1 \leq \Pdiv{\a^t}{\a^{t-1}}. \label{eq:sc1}
% \end{gather}
Since all coordinates of $\a$ are non-decreasing, the actual cost for
step $t$ is
\[
  \norm{\a^t - \a^{t-1}}_1 = \sum_i (\a^t_i - \a^{t-1}_i)
  \stackrel{(\text{\ref{eq:pmp}})}{\leq} \sum_i
  \a^t_i \log \frac{\a^t_i}{\a^{t-1}_i} = \Pdiv{\a^t}{\a^{t-1}}. \tag{$\star\star\star$}
\]

Summing the starred equations shows that in serving any request,
\[ (\text{cost for } ALG) + \Delta \Phi \leq (\log 1/\delta) \cdot
  (\text{cost for } OPT).
\]
Since $\Phi_0 = 0$ and $\Phi_t \geq 0$, we get the following theorem.

\begin{theorem}[Set Cover]
  For the online (unweighted) set cover problem, the projection
  algorithm maintains a fractional solution with at most
  $(\log n) \cdot OPT + 1$ sets, where $OPT$ is the optimal integer solution.
\end{theorem}

Observe that increasing the $\delta$ term improves the multiplicative
guarantee but worsens the additive term. By a guess-and-double approach,
we can get an $O(\log \frac{n}{OPT}) \cdot OPT$-approximation.
We chose the analysis for the unweighted case for its simplicity. To
extend to the case, e.g., where sets have costs, we need to consider
a slightly different Bregman divergence. We defer this
discussion to a later version of the paper.

%%% Local Variables:
%%% mode: latex
%%% TeX-master: "main"
%%% End:

\newcommand{\rootvtx}{\mathbf{r}}
\newcommand{\ch}{{\small\mathsf{ch}}}
\renewcommand{\d}{\delta}
\renewcommand{\t}[1]{\tilde{#1}}

\section{The $(h,k)$-Server Problem on Trees}
\label{sec:server}

We now consider the $(h,k)$-server problem on trees. For readers wishing
to gain more familiarity with projection-based algorithms, and analyses
using KKT conditions, we recommend \S\ref{sec:paging} for an analysis of
a projection-based weighted paging algorithm. However, reading
\S\ref{sec:paging} is optional for the experts, since the present section is self-contained.

Given a tree $T$, the root $\rootvtx$ is at depth $0$; the depth for a
node $u$ is the number of edges on the $\rootvtx$-$u$ path. We assume
all the leaves are at depth $D$; this is for convenience, and incurs no
significant loss of generality.  The vertices at depth $d$ are denoted
by $V_d$, and hence $V_0 = \{\rootvtx\}$, whereas $V_D$ is the set of
all leaves.  Let $w_u$ be weight of node $u$; we define the distance from $u$ to $v$ as the total sum of weights of vertices on the path from $u$ to $v$ (this is equivalent to defining suitable edge weights).
%Let $w_u$ be the length of the edge from $u$ to its parent.
Let $n = |V_D|$ be the number of
leaves of the tree, and we associate the leaves with the set $[n]$.
Without loss of generality we assume that requests only appear at the leaves.

Similarly to the paging problem, the goal in the weighted $(h,k)$-server problem
is to respond at each time $t$ to requests $r_t \in [n]$  by producing a vector
$\z^t \in \{0,1\}^n$, with $\norm{\z^t}_1 = k$, such that $\z^t_{r_t} =1$
for all $t$. The objective is to minimize the total weighted change:
\begin{gather}
  d(\z^t, \z^{t-1}) := \sum_{u \in V(T)} w_u \; | \z^t(T_u) -
  \z^{t-1}(T_u) |. \label{eq:kserver-distance}
\end{gather}
Here $T_u$ denotes the set of vertices in the subtree rooted at
$u$, and $\z^t(T_u)$ denotes the number of servers in this subtree
as defined by $\z^t$. We compare our performance to that of the optimal
solution that uses only $h$ servers.

\subsection{Atoms and the Anti-Server Polytope}

Following the paging setting, we consider an
``anti-server'' polytope, where an $x_u$-value of $0$ indicates that there is a
server at $u$, and $1$ otherwise. In particular, we use
the anti-server polytope proposed by Bubeck et al.~\cite{BCLLM17}, which has
many nice features. It will be crucial to define variables for both
leaves and internal nodes; in fact, each internal node may have
several variables corresponding to it, as we now explain.

Let $L_u$ be the set of leaves in the tree $T_u$. Each vertex $u \in V(T)$
has associated variables $x_{u,j}$ for $j \in \{1, 2, \ldots,
|L_u|\}$. We refer to pairs $(u,j)$ as \emph{atoms}. Let
$\chi_u := \{ (v, \ell) \mid v \text{ is a child of } u, \ell \in
[|L_v|]\}$ be the atoms corresponding to the children of $u$, and let
$\chi_u$ be the children of $u$ in the tree $T$. Since the
leaves are all at the same level, the total number of atoms (and hence
the number of variables) at each level of the tree is exactly $n$, the number of
leaves. Moreover, each leaf $u$ has only a single atom $(u,1)$ and hence a
single variable $x_{u,1}$. % As in the MTS case, define the ``shifted''
% variables $\xt_{u,j} := x_{u,j} + \delta$, where $\delta$ will be
% specified later.
Let $N = n(D+1)$ be the total number of atoms.

The \emph{anti-server polytope} proposed by Bubeck et al.~\cite{BCLLM17} is the following:
\begin{alignat*}{2}
P:= \quad \bigg\{ x \in [0,1]^N \;\Big| \qquad \qquad  x_{\rootvtx,j} &\geq \mathbf{1}_{(j > h)} & & \forall j\\
  \sum_{j \leq |S|} x_{u, j} &\le \sum_{(v,\ell) \in S} x_{v,\ell}
  & \qquad \qquad & \forall u, S \sse \chi_u \qquad\bigg\}
\end{alignat*}

Remember that $P$ represents solutions in the anti-server world as
follows: to encode an integer solution $B^t$, set $\y_{\ell,1} = 0$ for
leaves $\ell$ containing servers and $\y_{\ell,1} = 1$ otherwise; set $\y_{u,j} = 0$ if the subtree under $u$ contains at least $j$
servers, and $1$ otherwise. It is easy to check that the constraints are
satisfied by this integral solution.

The reader may find it convenient to think
of an ``ideal'' fractional solution for $P$ as follows: the ``ideal''
setting for a $\y_{u,\cdot}$ vector at an internal node $u$ is to take the
vectors of its children, concatenate them, and then sort the entries of
this concatenated vector in non-decreasing order. While such a sortedness
condition is not required, and may not even hold, it may be useful for intuition about the LP.

\subsubsection{Translating between Servers and Anti-Servers}

Define $\delta := \frac{k-h+\nicefrac12}{k+\nicefrac12}$, as for paging.
We define the shifted polytope
$P_\delta := P \cap \{ x \mid x_{(u,1)} \geq \delta \;\forall u \in
L_\rootvtx\}$ to be the subset of points such that all leaf atoms have
value at least $\delta > 0$. The algorithm maintains a fractional
solution $\x^t \in P_\delta$, with $\x^t_{r_t,1} = \delta$ and
$\norm{\x^t} = n-h$. We can define a norm on vectors in $\R^N$ as
follows:
\begin{gather}
  \Tlone{x} := \sum_u w_u \; \sum_{j \in [|T_u|]}
  |x_{u,j}|. \label{eq:treenorm}
\end{gather}
Note: the norm~(\ref{eq:treenorm}) is defined for vectors in $\R^N$ that
assign values to all atoms in the tree, whereas the
distance~(\ref{eq:kserver-distance}) is defined for vectors in $\R^n$
that assign values only to the leaves. However, the two distances can be
related to each other using Lemma~\ref{lem:emd-to-lone}.

Define the vector of servers $\z^t \in [0,1]^n$ by setting
$\z^t_u := \frac{\mathbf{1} - \x^t_{u,1}}{1-\delta}$ for each leaf
$u$. This gives a fractional vector with
$\norm{\z^t} = k + \nicefrac12$; Lemma~\ref{lem:emd-to-lone} shows that
\[ d(\z^t, \z^{t-1}) \leq \textstyle \frac{1}{1-\delta}\cdot \Tlone{\x^t
    - \x^{t-1}}. \] Finally, for hierarchically well-separated trees
(HSTs) we can use~\cite[Lemma~3.4]{BCLLM17} and \cite[\S5.2]{BBMN11} to
round the fractional solutions with $k+1/2$ servers to integer solutions
with just $k$ servers, so
that
\[ \E\big[\, d(\hat\z^t, \hat\z^{t-1}) \,\big] \leq O( d(\z^t,
\z^{t-1})). \]

The main result of this section is an algorithm to maintain a fractional
point $\x^t$ as follows.
\begin{theorem}[Main Theorem: $(h,k)$-server]
  \label{thm:main-kserver}
  Given a tree of depth $D$, there exists an algorithm that
  maintains fractional solutions $\x^t \in P_\delta$ with
  $\x^t_{r_t} = \delta$ and $\norm{\x^t} = n-h$, such that for any
  integer feasible solutions $\y^t$, we have
  \[ \Tlone{\x^t - \x^{t-1}} \leq O(D \log (1+\nicefrac1\delta)) \cdot \Tlone{ \y^t -
      \y^{t-1} } + C' = O(D\log (1+\nicefrac1\delta)) \cdot d(\y^t,
    \y^{t-1}) + C'. \]
  for an additive term $C'$ that depends only on $T,k,\delta$ but not on
  the input sequence.
\end{theorem}
Combining this theorem with the above chain of inequalities gives us a
randomized algorithm for the $(h,k)$-server problem with a competitive
ratio of
$O(D \, \frac{\log (1+ 1/\delta)}{1-\delta}) = O(D \log
\frac{2k-h+1}{k-h+\half})$, at least when $h \in \Omega(k)$. Note that
for $h=k$ the algorithm is $O(D\log k)$-competitive, and for $h=k/2$ the
algorithm is $O(D)$-competitive.  In the rest of this section, we
present the proof of Theorem~\ref{thm:main-kserver}.

\subsection{The Projection Algorithm}

When a request arrives at leaf $r_t$ at time $t$, we consider the polytope
\[ P_t := P \cap \{x_{r_t,1} \leq \delta\}. \] Given the previous
solution $\x^{t-1} \in P_{t-1} \cap P_\delta$, the new point $\x^t$ is the projection
of $\x^{t-1}$ onto $P_t$, with respect to the Bregman divergence for the
\emph{(shifted) multilevel entropy function}:
\[ \Ddiv{x}{x'} := \sum_u w_u \sum_j \bigg( \t{x}_{u,j} \log
  \frac{\t{x}_{u,j}}{\t{x}'_{u,j}} - \t{x}_{u,j} + \t{x}'_{u,j}
  \bigg). \] (Here the tildes over the variables denote an additive
shift by $\delta$, so that $\t{x} = (x+\d)$.)  In other words, we set
$\x^t := \arg\min_{x \in P_t} \Ddiv{x}{\x^{t-1}}$. Since
$P_t \neq P_\delta$, we need to show that $\x^t$ indeed lies in $P_\delta$;
this appears in Claim~\ref{clm:box}.

\subsubsection{The Optimality Conditions}
\label{sec:opt-cond}

The projection problem above can be written as:
\begin{alignat}{2}
  \min_{x} \sum_{u \neq r_t} w_u \sum_j & \bigg( \t{x}_{u,j} \log
    \frac{\t{x}_{u,j}}{\t{\x}^{t-1}_{u,j}} - \t{x}_{u,j} \; + \; \t{\x}^{t-1}_{u,j} \bigg)  \notag\\
s.t. \quad  x_{\rootvtx,j} &\geq \mathbf{1}_{(j > h)} & & \forall j \label{eq:ks-root} \\
  \tsty \sum_{j \leq |S|} x_{u, j} &\leq \tsty \sum_{(v,\ell) \in S} x_{v,\ell}
  & \quad & \forall \;\text{non-leaves } u, \forall \,S \sse \chi_u \label{eq:ks-subsets} \\
  x_{r_t,1} &\leq \delta \label{eq:ks-req} %\\
  %x_{u,j} &\geq 0 && \forall u \in V_1, \forall j
\end{alignat}
Each constraint of the
form~(\ref{eq:ks-subsets}) is uniquely specified by some set $S$ of atoms that
share a common parent, which we denote by $p(S)$ --- in other words,
$S \sse \chi_{p(S)}$.

The KKT optimality conditions show that for all
$u \not\in \{\rootvtx, r_t\}$,
\begin{gather}
  w_u \log \frac{\t\x^t_{u,j}}{\t\x^{t-1}_{u,j}} =
  \underbrace{\sum_{S \sse \chi_{p(u)}: (u,j) \in S} \lambda_{S}}_{a_{u,j}}  -
  \underbrace{\sum_{T \sse \chi_{u}: j \leq |T|}
    \lambda_{T}}_{b_{u,j}} %\; + \; \gamma_{u,j}.
  \tag{KKT2a} \label{eq:kkt-ks}
\end{gather}
where all the Lagrange multipliers are non-negative. While we omit the superscripts $t$ for the variables $\lambda, a_{u,j}, b_{u,j}$, we emphasize that all these terms are different for each time $t$.
Note that the
leaves $u$ have no $b_{u,j}$ terms, only the (positive) $a_{u,j}$
terms. Hence all non-$r_t$ leaves can only increase in value. (We
sometimes refer to $b_{u,j}$ for a leaf $u$, in which case imagine
$b_{u,j}=0$.)

For the root atoms, the KKT conditions have the dual variable
$\lambda_{\rootvtx,j} \geq 0$ corresponding to
constraint~(\ref{eq:ks-root}), but they do not have any $a_{\rootvtx, j}$
terms:
\begin{gather}
  w_\rootvtx \log \frac{\t\x^t_{\rootvtx,j}}{\t\x^{t-1}_{\rootvtx,j}} =
 \underbrace{\lambda_{\rootvtx,j}}_{a_{\rootvtx,j}} -
  \underbrace{\sum_{T \sse \chi_{\rootvtx}: j \leq |T|}
    \lambda_{T}}_{b_{\rootvtx,j}} .
  \tag{KKT2b} \label{eq:kktroot-ks}
\end{gather}
For the demanded vertex $r_t$, which has a single atom $(r_t, 1)$:
\begin{gather}
  w_{r_t} \log \frac{\t\x^t_{r_t,1}}{\t\x^{t-1}_{r_t,1}} =
  \underbrace{\sum_{S \sse \chi_{p(r_t)}: (r_t,1) \in S}
    \lambda_{S}}_{a_{r_t,1}} \; - \;\gamma_{t} \tag{KKT2c} \label{eq:kktprime-ks}
\end{gather}
where $\gamma_t \geq 0$ corresponds to
constraint~(\ref{eq:ks-req}). Again, we assume $b_{r_t,1} = 0$.
In the rest of the paper, we define
\begin{gather}
  A^t_{u,j} := \t\x^t_{u,j} a^t_{u,j} \qquad \text{and} \qquad B^t_{u,j} := \t\x^t_{u,j} b^t_{u,j}
\end{gather}

Finally, complementary slackness implies:
\begin{gather}
  \bigg( \lambda_{S} > 0 \implies \sum_{j \leq |S|} \x^t_{p(S), j} = \sum_{(v,\ell)
    \in S} \x^t_{v,\ell} \bigg)
  \iff \bigg( \lambda_S \sum_{j \leq |S|} \x^t_{p(S),j} = \lambda_S \x^t(S)
  \bigg). \tag{CS2} \label{eq:tight}
\end{gather}
where $p(S)$ is the ``parent'' node for the set of nodes in $S$.
The other two variables give us:
\begin{gather}
  \lambda_{\rootvtx, j} \cdot ( \x^t_{\rootvtx, j} - \mathbf{1}_{(j > h)})
  = 0, \label{eq:ks-csroot} \\
  \gamma_t \, \x^t_{r_t,1} = \gamma_t\, \delta . \label{eq:ks-req-cs}
\end{gather}

\subsection{Properties of the Projected Point}
\label{sec:prop-proj}

We prove some useful properties for the new optimal solution $\x^t$.
These properties are satisfied by $\x^0$ by construction, and we inductively assume
that they hold for $\x^{t-1}$, in order to prove them for $\x^t$.  We defer the
proofs until later; these are very similar to those
in~\cite{BCLLM17}.

\begin{restatable}[Root is Tight]{claim}{ksRoottight}
  \label{clm:root-vals}
  For the root vertex $\rootvtx$, $\x^t_{\rootvtx,j} = \mathbf{1}_{(j > h)}$.
\end{restatable}

\begin{restatable}[Box Constraints]{claim}{ksBox}
  \label{clm:box}
  For all $u, j$, $0 \leq \x^t_{u,j} \leq 1$. Moreover,
  $\x^t_{r_t,1} = \delta$. Finally,
  $\x^t_{u,1} \geq \delta$ for leaves $u$, and hence $\x^t \in P_\delta$.
\end{restatable}

\begin{restatable}[Flow]{lemma}{ksFlow}
  \label{lem:flow}
  For each internal node $u$,
  $\sum_j \x^t_{u,j} = \sum_{(v,\ell) \in \chi_u} \x^t_{v,\ell}$.
  This implies that for any depth $d$,
  \[ \tsty \sum_{u \in V_d} \sum_j \x^t_{u,j} = n-h. \]
  So the difference between $\x^{t-1}$ and $\x^t$ can be viewed as a
  flow from $r_t$ to the other leaves in $T$.
\end{restatable}

The following lemma, using complementary slackness, is crucial to
relate the dual values across levels.
\begin{restatable}[Relating Consecutive Levels]{lemma}{ksLevels}
  \label{lem:cs-useful}
  For any node $u$ in the tree
  \[ \sum_j B^t_{u,j}  = \sum_{(v,\ell) \in \chi_u} A^t_{v,\ell}
    = \sum_{T \sse \chi_u} \lambda_T \t\x^t(T). \] % Moreover, an
  % analogous bound holds with $x(\cdot)$ replaced by $\xt(\cdot)$.
\end{restatable}

\subsection{The Potential Function}

If $\y^t \in \{0,1\}^N$ is the optimal solution, and
$\x^t \in [0,1]^N$ is our solution, the potential is defined
as
\[ \Pdiv{\y^t}{\x^t} \quad = \quad \sum_u w_u \sum_j \t\y^t_{u,j} \log
  \frac{\t\y^t_{u,j}}{\t\x^t_{u,j}}
  %\quad = \quad \sum_u w_u \sum_{j: \y^t_{u,j} = 1} \log
  % \frac{1}{\x^t_{u,j}}
\]
Observe that each term in the inner sum lies in the range $[- \d \log (1+\nicefrac1\d), (1+\d) \log (1+\nicefrac1\d)]$. %$[-\frac{1+\d}{e}, (1+\d) \log (1+\nicefrac1\d)]$.

\subsubsection{When OPT Moves: Upper Bounding the Potential Gain }
\label{sec:when-opt-moves}

Suppose OPT moves from $\y^{t-1}$ to $\y^t$. % changing
% the potential from $\Pdiv{\y^{t-1}}{\x^{t-1}}$ to
% $\Pdiv{\y^t}{\x^{t-1}}$.
Changing a coordinate $\y_{u,j}$ from $0$ to $1$ causes OPT to pay
  $w_u$ for such an increase---recall that it only pays for increases,
  and not decreases. Moreover, the increase in potential is
\[ w_u ((1+\d) \log (1+\d) - \d \log \d - \log \t\x^{t-1}_{u,j}) \leq
  w_u (1+\d)\log (1+\nicefrac1\d).
  \] Moreover, changing $\y_{u,j}$ from $1$ to $0$ only decreases the
  potential. Hence, we get
\begin{lemma}
  \label{lem:opt-charge}
  $\Pdiv{\y^t}{\x^{t-1}} - \Pdiv{\y^{t-1}}{\x^{t-1}} \leq (1+\d) \log (1+\nicefrac1\delta)
  \cdot \Delta OPT$.
\end{lemma}

An aside: to see why the variables in the Bregman divergence, and hence
the potential, are shifted by $\delta$, observe that we do not ensure
that all $\x_{u,j}$ variables are at least $\delta$---only the leaves
are at least $\delta$. However, to prove the above lemma, we need to
control the potential change and make the gradients Lipschitz even at
the non-leaf nodes, so the terms are shifted explicitly by $\delta$.

\subsubsection{When ALG Moves: Lower Bounding the Potential Drop}
\label{sec:when-alg-moves}

Next, the algorithm changes its solution from $\x^{t-1}$ to $\x^t$, and
the rest of the argument will be to bound the amortized cost.
Since we use a Bregman divergence
to project a point $\x^{t-1}$ down to $\x^t \in P_t$ (and also because
$\y^t \in P_t$), the reverse-Pythagorean property implies:
\begin{gather}
  \Ddiv{\y^t}{\x^{t-1}} \geq \Ddiv{\y^t}{\x^t} + \Ddiv{\x^t}{\x^{t-1}}
  \notag  \\
  \implies \underbrace{\Pdiv{\x^t}{\x^{t-1}}}_{\text{``shadow'' cost}} +
  \underbrace{\big(\Pdiv{\y^t}{\x^t} - \Pdiv{\y^t}{\x^{t-1}} \big)}_{\text{change
      in potential}} \leq 0.
\end{gather}
We now bound the algorithm's movement cost by some small factor times
this ``shadow'' cost $\Pdiv{\x^t}{\x^{t-1}}$.  Substituting the
expressions for $w_u \log \frac{\t\x^t_{u,j}}{\t\x^{t-1}_{u,j}}$
from~(\ref{eq:kkt-ks}-\ref{eq:kktprime-ks}) into the definition of
$\Pdiv{\cdot}{\cdot}$, and observing that
$\gamma_t \x^t_{r_t, 1} = \delta \gamma_t$, we get that the ``shadow
cost'' is:
\begin{align}
  \Pdiv{\x^t}{\x^{t-1}}
  &= \sum_{u} w_u \sum_j \t\x^t_{u,j} \log
    \frac{\t\x^t_{u,j}}{\t\x^{t-1}_{u,j}} \notag \\
  &=
    \sum_{u} \sum_j \t\x^t_{u,j} (a_{u,j} - b_{u,j}) -
    2\delta\gamma_t =
    \sum_{u} \sum_j (A^t_{u,j} - B^t_{u,j}) -
    2\delta\gamma_t
  \label{eq:1}
\end{align}

Recall that the leaf nodes (i.e., nodes at depth $D$) do not have any $b_{u,j}$ terms
in~(\ref{eq:kkt-ks}). Lemma~\ref{lem:cs-useful} now allows us to
cancel the $B^t_{u,j}$ with the $A^t_{v,\ell}$
terms on the next level, giving:
\begin{align}
 \Pdiv{\y^t}{\x^{t-1}} - \Pdiv{\y^t}{\x^t} \geq  \Pdiv{\x^t}{\x^{t-1}}
% &= \sum_{u \neq root} \sum_j \x^t_{u,j} A_{u,j} - \sum_{u
%                  \neq root} \sum_{(v,\ell) \in \chi_u} \x^t_{v,\ell}
%                  A_{v,\ell} - \delta\gamma \notag\\
               &= \sum_j A^t_{\rootvtx,j} - 2\delta\gamma_t.
               % =  \sum_{S \sse \chi_{\rootvtx}}
               %   \lambda_S \x^t(S) - \delta\gamma
                \label{eq:7}
\end{align}
It now suffices to bound the movement cost of the algorithm by some
constant factor times the expression in~(\ref{eq:7}). We will not manage
to do that; instead we give another lower bound on the drop in
potential. For some $\x$, let $W(\x) := \sum_u \sum_j w_u \x_{u,j}$.
Observe that $\Ddiv{\y}{\x} = \Pdiv{\y}{\x} - W(\y) + W(\x)$, and hence
the $W(\cdot)$ function captures the difference between $\Ddiv{\cdot}{\cdot}$ and $\Pdiv{\cdot}{\cdot}$.

\begin{restatable}[Second Lower Bound]{lemma}{lemSecond}
  \label{lem:second-bound}
  $\Ddiv{\y^t}{\x^{t-1}} - \Ddiv{\y^t}{\x^t} \geq \delta\gamma_t$. Thus,
  \[ \Pdiv{\y^t}{\x^{t-1}} - \Pdiv{\y^t}{\x^t} \geq \delta\gamma_t +
    W(\x^t) - W(\x^{t-1}). \]
\end{restatable}

We defer the proof to \S\ref{sec:proof-secondbound}. Averaging the expression in Lemma~\ref{lem:second-bound} with~(\ref{eq:7}) gives us
\begin{gather}
  \Pdiv{\y^t}{\x^{t-1}} - \Pdiv{\y^t}{\x^t} \geq \nicefrac1{3} \sum_{j}
  A^t_{\rootvtx,j} +
  \nicefrac23 (W(\x^t) - W(\x^{t-1})). \label{eq:ks-shadow}
\end{gather}
The linear terms $W(\x^t) - W(\x^{t-1})$ will telescope over time, and hence the
interesting term is the summation $\sum_{j}
  A^t_{\rootvtx,j}$.
In the next section, we relate the movement cost of the algorithm to this
summation, which will complete the argument.

\subsection{Bounding the Movement by the Shadow Cost: Shallow Trees}

We now bound the movement cost
$\sum_t \sum_{u} w_u \sum_j (\x^t_{u,j} - \x^{t-1}_{u,j})^+$ for the
entire sequence.\footnote{Note that it suffices to bound the increase in
coordinates, since the total movement is at most twice this amount, plus
an additive constant that depends only on the instance and is
independent of the request sequence.}
First, let us record a simple observation.
\begin{lemma}
  \label{lem:bound1}
  Suppose we have values $y, y' \geq 0$ such that
  $c \log \frac{\t{y}}{\t{y}'} \leq (a-b)$ with some $a, b \geq 0$ and $c > 0$.
  Then
  \begin{gather*}
    c\cdot (y - y')^+ \leq \t{y}\cdot a. %\label{eq:bound1}
  \end{gather*}
\end{lemma}

\begin{proof}
  If $y \leq y'$ then $c\cdot (y-y')^+ = 0 \leq
  \t{y}\cdot a$. Else, when $y > y'$,
  \[ c \cdot (y-y')^+ = c\cdot (y-y') = c\cdot (\t{y}-\t{y}') \stackrel{(\ref{eq:pmp})}{\leq} c \cdot \t{y} \log
  \frac{\t{y}}{\t{y}'} \leq
  \t{y}(a-b) \leq \t{y}\cdot a. \qedhere\]  %This proves~(\ref{eq:bound1}).
\end{proof}

Using~(\ref{eq:kkt-ks}-\ref{eq:kktprime-ks}) in conjunction with
Lemma~\ref{lem:bound1} (where we set $a = a^t_{u,j}, b = b^t_{u,j}, c = w_u$),
\begin{align}
  \sum_{u} w_u \sum_j (\x^t_{u,j} - \x^{t-1}_{u,j})^+
  &\leq
  % \sum_{u} \sum_j (A^t_{u,j} + B^t_{u,j})
  %   \stackrel{(\text{Lemma~\ref{lem:cs-useful}})}{\leq}
    \sum_{u} \sum_j A^t_{u,j}.  \label{eq:5}
  %   \notag \\
  % &= \sum_{h=1}^H \sum_{v \in V_h} \sum_j \x^t_{u,j} (2A_{u,j} + \gamma_{u,j})
\end{align}
We do not get the $\gamma_{t}$ term, because this corresponds to the
requested node $r_t$; since the $\x^t_{r_t}$ value decreases, the
corresponding $(\x^t_{r_t,1} - \x^{t-1}_{r_t,1})^+$ term is in fact
zeroed out.

Since the lower bound on OPT is just in terms of the
$A^t_{\rootvtx,j}$ terms, we want to argue that all the non-root terms
in~(\ref{eq:5}) are bounded by the terms corresponding to $\rootvtx$. That is almost what
we now show (modulo a certain additive term that telescopes over time). For
brevity, given the $\x^t_{u,j}$ values, define
\begin{gather}
  \x^t_u := \sum_j \x^t_{u,j} . \label{eq:agg}
\end{gather}

\begin{lemma}
  \label{lem:aux-claim}
  For each non-leaf node $u$, we have
  \begin{gather}
    \sum_{(v,\ell) \in \chi_u} A^t_{v,\ell} \leq \sum_j A^t_{u,j} -
    w_u(\x^{t}_u - \x^{t-1}_u). \label{eq:charge-up}
  \end{gather}
\end{lemma}

\begin{proof}
  We focus on a node $u$ that is neither a leaf nor the root; the root case
  is very similar. Non-negativity of Bregman divergences (or
  equivalently,~(\ref{eq:pmp})) and $w_u \geq 0$ imply
  \[ w_u \left( \t\x^t_{u,j} \log \frac{\t\x^t_{u,j}}{\t\x^{t-1}_{u,j}} -
      \t\x^t_{u,j} + \t\x^{t-1}_{u,j} \right) \geq 0. \]
  Now applying~(\ref{eq:kkt-ks}), using definition~(\ref{eq:agg}) and
  cancelling the additive $\delta$ terms gives \[ (A^t_{u,j}
  - B^t_{u,j}) - w_u(\x^t_{u,j} - \x^{t-1}_{u,j}) \geq 0. \] Finally, summing up
  over all $j$, and using Lemma~\ref{lem:cs-useful}
  to replace $\sum_j B^t_{u,j}$ by $\sum_{(v,\ell) \in \chi_u}
  A^t_{v,\ell}$ completes the proof. The proofs for the root $\rootvtx$
  is similar,
  using~(\ref{eq:kktroot-ks}) instead. For
  the root $\rootvtx$, observe that the total mass does not change, so the
  term $w_\rootvtx(\x^t_\rootvtx - \x^{t-1}_\rootvtx) = 0$.
\end{proof}

Now we can multiply~(\ref{eq:charge-up}) by $(D-d)$ for vertices
$u \in V_d$ with $d = 0,1,2,\ldots,D-1$, sum these up, and add $\sum_j
A^t_{\rootvtx, j}$ to both sides to get
\begin{align*}
  \sum_{u,j} A^t_{u,j} \leq  (D+1) \sum_j A^t_{\rootvtx,j} - \sum_{d =
  1}^{D-1} \sum_{u \in V_d} (D-d) \cdot w_u(\x^t_{u} - \x^{t-1}_{u}).
\end{align*}
Summing up over all times $t$ and using~(\ref{eq:5}) gives us
\begin{align*}
  ALG &\leq (D+1) \sum_{t=1}^T\sum_{j} A^t_{\rootvtx,j} - \sum_{d =
  1}^{D-1} \sum_{u \in V_d} (D-d) \cdot w_u(\x^T_{u} - \x^0_{u}) \\  &\leq
  (D+1) \sum_{t=1}^T\sum_{j} A^t_{\rootvtx,j} +
  D (W(\x^0) - W(\x^T)).
\end{align*}
Combining with~(\ref{eq:ks-shadow}), this implies
\[ ALG  
  \leq 3(D+1) \sum_t \Big( \Pdiv{\y^t}{\x^{t-1}} - \Pdiv{\y^t}{\x^t} 
  \Big) + O(D)[W(\x^0) - W(\x^T)]. \]
Since $|\x^T_u - \x^0_u| \leq h$ for each node $u$, where $h$ is
the number of servers that the optimal algorithm has in the $(h,k)$
server problem, $W(\x^T) - W(\x^0) \leq h \sum_u w_u$. Now using Lemma~\ref{lem:opt-charge} to bound the change in potential due to OPT, we get
\[ ALG 
  \leq O(D \log (1+ \nicefrac1\delta)) \cdot OPT + 3(D+1)\Big(
  \Pdiv{\y^0}{\x^0} - \Pdiv{\y^T}{\x^T} \Big) + O(Dh) \sum_u w_u. \]
This proves Theorem~\ref{thm:main-kserver} with the additive term
$C' = O(Dh \sum_u w_u + D\cdot (\Pdiv{\y^0}{\x^0} - \Pdiv{\y^T}{\x^T})$.

\subsubsection{An $O(\log n \log k)$-competitive algorithm for HSTs}

Let us focus on the $k$-server problem; the extensions to $(h,k)$-server
are immediate. % Moreover, we assume that the HSTs are $\alpha$-HSTs with
% $\alpha = O(1)$; i.e., the edge lengths drop by a factor of at least
% $\alpha$ as we walk down from the root to the leaves.
For the $k$-server problem on HSTs,
Theorem~\ref{thm:main-kserver} implies an
$O(\log \Delta \log k)$-competitive algorithm, by setting
$\delta = \frac{1}{2k+1}$, and using the fact that the depth of any
HST is $O(\log \Delta)$. Here $\Delta$ is the \emph{aspect ratio} of the
tree, the ratio of the largest to smallest distance in the tree.

To get the improved result of $O(\log n \log k)$, we simply use the fact
that for any HST, there is another tree with depth $O(\log n)$ which
changes distances by at most a constant factor (see,
\cite[Theorem~5.1]{BBMN11} for a formal statement). The basic idea for
obtaining this tree is simple: for each vertex $u$, if it contains a
child $v$ such that $|L_v| \geq |L_u|/2$, i.e., the number of leaves
under $v$ is at least half the number under $u$, then we contract the
edge $(u,v)$, and make the new node have weight equal to the parent's
weight. This ensures that traversing each edge reduces the number
of leaves by a factor of $2$ and hence gives a tree with depth
$O(\log n)$; moreover it does not change distances by more than a constant
factor. This implies the following:

\begin{theorem}
  There is an $O(\log n \log k)$-competitive randomized algorithm for
  the $k$-server problem on HSTs.
\end{theorem}

In the next section, we improve this bound to get $O(\log^2 k)$-competitiveness.

%%% Local Variables:
%%% mode: latex
%%% TeX-master: "main"
%%% End:

\newcommand{\apot}{\Psi_2}

%\newpage
\section{An $O(\log^2 k)$ Bound for $k$-server}
\label{sec:logsq}

We now give the proof of Theorem~\ref{thm:logsq}. The proof here is
somewhat longer and more involved than in~\cite{BCLLM17}---while it can
conceivably be shortened, we currently believe that some of the
complexity is due to the algorithm being defined as a sequence of
discrete jumps, rather than via a continuous trajectory. That being
said, the high-level idea of the proof is simple and modular (and
parallels that in~\cite{BCLLM17}).

Recall that for a leaf $u$ we defined
$\z^t_u := \frac{1 - \x^t_{u,1}}{1-\delta}$. We extend this definition
for an internal node $u$ as $\z^t_u := \sum_{v \in leaves(T_u)}
\z^t_v$. Define $\|\mathbf{v}\|^+_w := \sum_{u} w_u
\mathbf{v}_u^+$. The main result of this section is the following:
\begin{theorem}
  \label{thm:logsq-tech}
  Let $T$ be a $\tau$-HST, where $0 < \tau \leq 1/10$. Then
  there exists a potential function $\Psi$ such that for each time $t$,
  \begin{gather*}
    \| \x^t - \x^{t-1} \|_w^+ \leq O(\log (k/\delta)) \cdot \sum_j
    A^t_{\rootvtx,j} + \Psi(\z^t) - \Psi(\z^{t-1}).
  \end{gather*}
  Moreover, $\Psi(\z^T) - \Psi(\z^0) = O(\sum_u \frac{w_u}{1-\delta} (k +
  \log\frac{k}{\delta}))$.
\end{theorem}
The above left hand side is the algorithm's (positive)
movement. Combining with~(\ref{eq:ks-shadow}) and
Lemma~\ref{lem:opt-charge}, the same arguments give
\[ ALG \leq O(\log \nicefrac{k}{\delta} \log \nicefrac1\delta) \cdot OPT
  + \Psi(\z^T) - \Psi(\z^0) + O(\log
  \nicefrac{k}{\delta})(\Pdiv{\y^0}{\x^0} - \Pdiv{\y^T}{\x^T} + O(h \sum_u
  w_u). \]
Setting $\delta = 1/k$ gives the $O(\log^2 k)$-competitiveness, and
hence the  proof of Theorem~\ref{thm:logsq}.

\subsection{Proof of Theorem~\ref{thm:logsq-tech}}

The proof contains many ingredients in common with that
of~\cite{BCLLM17}, but the projection-based approach means we need some
further ideas (such as the potential $\apot$ below). As in their work, we show the result in two steps. We
first define values $\alpha^t_u \geq 0$ for each vertex. Let $q^t_u :=
\alpha^t_{p(u)} - \alpha^t_u \geq 0$ for each non-root vertex, and
$q^t_\rootvtx := 0$. For a given vector $q$, define $\|\mathbf{v}\|^+_q
:= \sum_{u} q_u w_u \mathbf{v}_u^+$.
\begin{lemma}
  \label{lem:logsq-part1}
  There exists choices of $\alpha_u^t$ and a potential function $\Psi_1$ such that for each time $t$,
  \begin{gather*}
    \| \x^t - \x^{t-1} \|_{\alert{q^t}}^+ \leq O(\log (k/\delta)) \cdot \sum_j
    A^t_{\rootvtx,j} + \Psi_1(\z^t) - \Psi_1(\z^{t-1}).
  \end{gather*}
\end{lemma}
Observe that this lemma bounds $\|\cdot\|_{q^t}$ rather than $\|\cdot\|_w$,
so we relate these two norms next:
\begin{lemma}
  \label{lem:logsq-part2}
  For the choice of $\alpha_u^t$ from Lemma~\ref{lem:logsq-part1}, there
  exist universal constants $c, c'$ and another potential function $\apot$
  such that for each time $t$,
  \begin{gather*}
    \| \x^t - \x^{t-1} \|_{\alert{w}}^+ \leq c\cdot \left( c'\cdot \| \x^t - \x^{t-1}
    \|_{\alert{q^t}}^+ - 2[\apot(\z^t) - \apot(\z^{t-1})] \right).
  \end{gather*}
\end{lemma}
Combining these two results and setting $\Psi(\z) = c \cdot c' \cdot \Psi_1(\z) -
2c \cdot \apot(\z)$  immediately gives
Theorem~\ref{thm:logsq-tech}. The proofs of these two lemmas appear
in the following sections.
% \agnote{Changed sign of potential above,
%   pls check and  remove this comment}

\subsubsection{Proof of Lemma~\ref{lem:logsq-part1}}

Recall $\x^t_u := \sum_j \x^t_{u,j}$. Following the proof of \eqref{eq:5}, using~(\ref{eq:kkt-ks}-\ref{eq:kktprime-ks}) in conjunction with
Lemma~\ref{lem:bound1}, % from \S\ref{sec:when-alg-moves},
\begin{gather}
  \sum_u w_u q^t_u (\x^t_u - \x^{t-1}_u)^+ \leq \sum_u w_u q^t_u \sum_j
  (\x^t_{u,j} - \x^{t-1}_{u,j})^+
  \le \sum_u q^t_u \sum_j
  A^t_{u,j}. \label{eq:3}
\end{gather}
The non-negativity of Bregman divergences and the fact that $\alpha^t_u
\geq 0$ implies
\begin{align}
  0 &\leq  w_u \alpha_u^t \sum_{j} \left( \t\x_{u,j}^t
    \log \frac{\t\x_{u,j}^{t}}{\t\x_{u,j}^{t-1}} - \t\x_{u,j}^t + \t\x_{u,j}^{t-1}
           \right)  \notag \\
  &\leq \alpha_u^t \sum_{j} \left(
    A^t_{u,j} - B^t_{u,j} \right) +  \alpha^t_u w_u \sum_j (\x_{u,j}^{t-1} -
    \x_{u,j}^{t})  \tag{by~(\ref{eq:kkt-ks})-(\ref{eq:kktprime-ks})}
    \notag \\
  &= \alpha_u^t \sum_{j}
    A^t_{u,j} - \alpha_u^t \sum_{(v,\ell) \in \chi_u}
    A^t_{v,\ell} +  \alpha^t_u w_u (\x_{u}^{t-1} -
    \x_{u}^{t}). \tag{by Lemma~\ref{lem:cs-useful}} \notag \\
  \intertext{Summing over all
  $u$, using $\x^t_u - \x^{t-1}_u = \frac{\z^{t-1}_u - \z^t_u}{1-\delta}$, and rearranging,}
  &
     \sum_{u \neq \rootvtx} \sum_{j} A^t_{u,j} (\alpha_{p(u)}^t -
    \alpha_u^t) \leq  \sum_{j} \alpha_\rootvtx^t A^t_{\rootvtx,j}  +
    \frac{1}{1-\delta} \sum_{u} \alpha^t_u w_u (\z_{u}^{t} -
    \z_{u}^{t-1}). \label{eq:4}
  % &\implies
  %    \sum_{u \neq \rootvtx} \sum_{j} x_{u,j}^t A^t_{u,j} (\alpha_{p(u)}^t -
  %   \alpha_u^t) \leq  \sum_{j} \alpha_\rootvtx^t x_{\rootvtx,j}^t A^t_{\rootvtx,j}  + \sum_{u} \alpha^t_u w_u (\alert{z_{u}^{t} -
  %   z_{u}^{t-1}})
\end{align}
Since $q_u^t = \alpha_{p(u)}^t - \alpha_u^t$, the left hand side of~(\ref{eq:4})
equals the right hand side of~(\ref{eq:3}). Therefore, to prove
Lemma~\ref{lem:logsq-part1} we need to:
\begin{OneLiners}
\item[(i)] choose $\alpha_u^t$ so that $\alpha_\rootvtx^t \leq O(\log
  k/\delta)$ and $q^t_u \geq 0$; and
\item[(ii)] choose a
  potential $\Psi_1(\cdot)$ such that
  $\frac{1}{1-\delta} \sum_{u}
  \alpha^t_u w_u (\z_{u}^{t} - \z_{u}^{t-1}) \leq \Psi_1(\z^t) - \Psi_1(\z^{t-1})$.
\end{OneLiners}
(Of course, we want this choice of $\alpha_u^t$ to
allow us to prove Lemma~\ref{lem:logsq-part2} as well.)

To this end, we define
\begin{align}
  \alpha_u^t & := \frac{1}{\z_u^t-\z_u^{t-1}}\int_{\z_{u}^{t-1}}^{\z_u^t}\ln\big(1+\frac{z}{\delta}\big)\,
               dz. \label{eq:alpha}
  % \\ &=
  %              \left.\frac{1}{\z_u^t-\z_u^{t-1}}\left((x+\delta)\ln(1+\frac{x}{\delta})-x\right)\right|_{\z_{u}^{t-1}}^{\z_u^t} \notag
  % \\
  %            & = \frac{(\z_u^t+\delta)\ln(1+\frac{\z_u^t}{\delta}) -
  %              (\z_u^{t-1}+\delta)\ln(1+\frac{\z_u^{t-1}}{\delta})}{\z_u^t-\z_u^{t-1}}-1 \notag
\end{align}
Since $0\leq \z_u^t \leq k$, the value
$\alpha_u^t \in [0, \ln(1+\frac{k}{\delta}) ]$. Moreover,
$\alpha_{p(u)}^t \geq \alpha_u^t$, since $\alpha_{(.)}^t$ is the average
value of an increasing function between two endpoints, and the
corresponding endpoints in $\alpha_{p(u)}^t$ are at least those in $\alpha_{u}^t$ (since $\z^{t-1}_{p(u)} \ge \z^{t-1}_u$ and $\z^t_{p(u)} \ge \z^t_u$).
% since the parent takes the average of
% the monotone function $\ln(1 + x/\delta)$ over the interval
% $[\z_{p(u)}^{t-1}, \z_{p(u)}^{t}]$ where $\z_{p(u)}^{t-1} \geq
% \z_{u}^{t-1}$ and $\z_{p(u)}^{t-1} \geq \z_{u}^{t-1}$.
Thus, the first condition above is satisfied.

Satisfying the potential condition is easy. Define
$$ \Psi_1(\z) = \sum_u \frac{w_u}{1-\delta}
\int_{\z_{u}^{0}}^{\z_u}\ln(1+\frac{z}{\delta})\, dz.$$ Since
$$\alpha^t_u (\z_{u}^{t} - \z_{u}^{t-1}) =
\int_{\z_{u}^{t-1}}^{\z_u^t}\ln(1+\frac{z}{\delta})\, dz$$ we
immediately get
$$\Psi_1(\z^{t-1}) + \sum_u \frac{w_u}{1-\delta} \alpha^t_u (\z_{u}^{t} -
\z_{u}^{t-1}) = \Psi_1(\z^{t}),$$ as we want. Putting these together
proves Lemma~\ref{lem:logsq-part1}.

\subsubsection{Proof of Lemma~\ref{lem:logsq-part2}}

\renewcommand{\top}{{\small\mathsf{top}}}

We now turn our attention to Lemma~\ref{lem:logsq-part2} and show the
choice of $\alpha^t$ from~(\ref{eq:alpha}) suffices. We consider
  a $\tau$-HST, for $\tau \leq 1/10$. We consider the potential
function $\apot(\z) := \sum_u w_u (\z_u - (2/3) \z_{p(u)})^+$, and want
to show
\begin{align}
	\frac{1}{c}\,\|\z^{t-1} - \z^t\|_w^+ \le c' \cdot \|\z^{t-1} -
\z^t\|_{q^t}^+ - 2 \left[\apot(\z^t) - \apot(\z^{t-1}) \right]. \label{eq:needPart2}
\end{align}

\begin{observation} \label{obs:subadd}
  For any $x,y$, $(x + y)^+ \le x^+ + y^+$.
\end{observation}

The \emph{backbone} (at time $t$) is the path from the root $\rootvtx$
to the requested node $r_t$. By the definition of $\z^t_u$ for internal
nodes $u$, we can visualize the difference between $\z^{t-1}$ and $\z^t$
as a \emph{flow} over $T$, where a total $1 - \z^{t-1}_{r_t}$ amount of
flow is sent from the non-$r_t$ leaves to $r_t$. For each leaf
$u \neq r_t$, there is a flow path $P_u$, and $\z^{t-1}_u - \z^{t}_u$
flow is sent along this path, reducing the $\z$-value at $u$ and
increasing it at $r_t$; the $\z$ values at internal nodes are obtained
by summing over all leaves in their subtree.
To track the difference between solutions $\z^{t-1}$ and $\z^t$, we
introduce an intermediary solution $\z'$ obtained by sending just the
flows from leaves lying within a ``$\z^{t-1}$-light'' subtree in the
backbone. 

More precisely, let $a$ be the highest node on the
backbone, where 
$$\z^{t-1}(T_{a}) \le \frac{1}{10}.$$ 
If $a = r_t$, or no
such node exists, simply define $\z' = \z^{t-1}$. Else, consider the flow
defined above, and let $\z'$ be obtained by applying to $\z^{t-1}$ all
the flows on paths contained within $T_a$. Since all these flows stay
within the subtree $T_a$, the value of its root $a$ remains
unchanged---i.e., $\z'_a = \z^{t-1}_a$.

It is technically simpler to track the change from the intermediate
solution $\z'$ to $\z^t$, instead of tracking it from $\z^{t-1}$ to $\z^t$. By the
next lemma, the $\|.\|^+_w$ movement from $\z^{t-1}$ to $\z^t$ can be
bounded in terms of the movement starting from the intermediate solution
$\z'$.

\begin{lemma} \label{lemma:first} There is a constant $c$ such that
  \begin{gather}
    \|\z^{t-1} - \z^t\|^+_w \le c \cdot \left(\|\z' - \z^t\|^+_w - 2
      \left[\apot(\z') - \apot(\z^{t-1})\right]\right). \label{eq:6}
  \end{gather}
\end{lemma}
	
\begin{proof}[Proof]
  If $\z'= \z^{t-1}$ the claim is vacuous, so assume the backbone node
  $a$ exists. By Observation~\ref{obs:subadd},
  \begin{align}
    \|\z^{t-1} - \z^t\|^+_w \le \|\z^{t-1} - \z'\|^+_w + \|\z' - \z^t\|^+_w. \label{eq:triangle}
  \end{align}
		
  We first claim that $\|\z' - \z^t\|^+_w \ge (9/10) w_a$, i.e., the
  remaining flow is large after we apply the flow paths contained in $T_a$. To see this,
  notice $\z'_{r_t} \le \z'_a = \z^{t-1}_a \le \frac{1}{10}$, and by
  feasibility of $\z^t$ we have $\z^t_{r_t} = 1$. Thus, to move from
  solution $\z'$ to $\z^t$ at least $9/10$ units of flow need to be sent
  into $r_t$. Moreover, since all this flow comes from leaves outside
  $T_a$, each unit of this flow pays at least $w_a$ when it ``enters the
  backbone''.
		
  Next, we claim $\|\z^{t-1}-\z'\|^+_w \le (2/90) w_a$. Using Observation \ref{obs:subadd},
  \begin{align}
    \|\z^{t-1}-\z'\|^+_w &= \sum_{u \in T_a \setminus a} w_u (\z^{t-1}_u - \z'_u)^+ \le \sum_{u \in T_a \setminus a} w_u (\z^{t-1}_u + \z'_u) \notag\\
    & = w_a \sum_{\ell \ge 1} \tau^{\ell} \left[\sum_{u \in \textrm{level $\ell$ of $T_a$}} \z^{t-1}_u + \sum_{u \in \textrm{level $\ell$ of $T_a$}} \z'_u\right] \notag\\
    & = w_a \sum_{\ell \ge 1} \tau^{\ell} (\z^{t-1}_a + \z'_a) \notag\\
			& \le (\nicefrac{2}{90}) w_a, \label{eq:first1}
  \end{align}
  where the last inequality uses the fact that we have a $\tau$-HST with $\tau \le 1/10$, and that by the definition of $a$ we have $\z^{t-1}_a = \z'_a \le 1/10$. %\mmnote{The third equality uses Lemma \ref{lem:flow} applied to $\z^{t-1}$ and the fact $\z'$ also satisfies it}
  By the previous paragraph, we now get $\|\z^{t-1}-\z'\|^+_w \le \frac{2}{81} \|\z'- \z^t\|_w^+$.

  Finally, we claim $\apot(\z') - \apot(\z^{t-1}) \le (4/90)  w_a$ (which is $\le \frac{4}{81} \|\z'-\z^t\|_w^+$): the subadditivity of $(.)^+$ implies subadditivity of $\apot$, hence (let $\y := \z' - \z^{t-1}$)
  \begin{align*}
    \apot(\z') - \apot(\z^{t-1}) &\le \apot(\z' - \z^{t-1}) = \sum_{u \in T_a \setminus a} w_u (\y_u - (2/3) \y_{p(u)})^+ \\
    & \le \sum_{u \in T_a \setminus a} w_u \left(\y_u^+ + (2/3) \y_{p(u)}^+\right) \le 2 \sum_{u \in T_a \setminus a} w_u \y_u^+,
  \end{align*}
  where the last inequality uses the fact $\y_a = 0$. Again by
  Observation \ref{obs:subadd} $\y_u^+ \le \z'_u + \z^{t-1}_u$, so part of
  inequality \eqref{eq:first1} gives $\sum_{u \in T_a \setminus a} w_u
  \y_u^+ \le (2/90) w_a$. This proves the claim. Moreover,
  % Using this last bound we can lower bound the right-hand side
  % in~(\ref{eq:6}) as:
  %
  \begin{align*}
    RHS(\ref{eq:6}) = c \cdot \left(\|\z'-\z^t\|_w^+ - 2 \left[\apot(\z') - \apot(\z^{t-1}) \right]\right) \ge c \cdot (\nicefrac{73}{81}) \|\z'-\z^t\|_w^+.
  \end{align*}
  By~\eqref{eq:triangle} and the consequence of \eqref{eq:first1}, %we can upper-bound the left-hand side in the statement of the lemma as
  $LHS(\ref{eq:6}) = \|\z^{t-1} - \z^t\|_w^+ \le (1 + \frac{2}{81})
  \|\z' - \z^t\|_w^+$. Now we have $LHS \le RHS$ for $c \ge \nicefrac{73}{83}$, concluding the proof.
\end{proof}

	Next we track the changes from the intermediate solution $\z'$ to $\z^t$. Since the value of $q^t$ is changing during this process, we need some notation to track it carefully. Given two solutions $\bar{\z}$ and $\bar{\z}'$, we define 
$$ q(\bar{\z}',\bar{\z})_u := \alpha(\bar{\z}'_{p(u)}, \bar{\z}_{p(u)}) - \alpha(\bar{\z}'_{u}, \bar{\z}_{u}),$$ 
where $$\alpha(z',z) := \frac{1}{z'-z} \int^{z'}_z \ln(1 + \nicefrac{x}{\delta})dx.$$ Notice that $q^t_u = q(\z^{t-1},\z^t)_u.$ The main technical part of this section will be to prove the following.
\begin{lemma}
	The following holds:
  \label{lem:main-tech}
  \begin{align}
    \|\z'-\z^t\|_w^+ \le c \cdot \|\z' -\z^t\|_{q(\z',\z^t)}^+ - 2
    \left[\apot(\z^t) - \apot(\z') \right]. \label{eq:need}
  \end{align}
\end{lemma}

Before we prove Lemma~\ref{lem:main-tech}, let us prove the main result using the above lemmas.

%\begin{lemma}
%  \label{lemma:second} $\frac{1}{c}\,\|\z^{t-1} - \z^t\|_w^+ \le c' \cdot \|\z^{t-1} - \z^t\|_{q(\z^{t-1},\z^t)}^+ - 2 \left[\apot(\z^t) - \apot(\z^{t-1}) \right].$
%\end{lemma}
	
\begin{proof}[Proof of Lemma \ref{lem:logsq-part2}]
	We prove inequality \eqref{eq:needPart2}. Putting Lemmas~\ref{lemma:first} and~\ref{lem:main-tech} together we
  get
  \begin{align*}
    \frac{1}{c}\,\|\z^{t-1} - \z^t\|_w^+ \le c' \cdot \|\z' - \z^t\|_{q(\z',\z^t)}^+ - 2 \left[\apot(\z^t) - \apot(\z^{t-1}) \right].
  \end{align*}
  This is almost what we wanted to prove, except that we have
  $\|\z' - \z^t\|_{q(\z',\z^t)}^+$ instead of
  $\|\z^{t-1} - \z^t\|_{q(\z^{t-1},\z^t)}^+$ on the RHS. But this is easily
  handled. First we change the $q$ and claim that
  $\|\z' - \z^t\|_{q(\z',\z^t)}^+ = \|\z' -
  \z^t\|_{q(\z^{t-1},\z^t)}^+$. Indeed, $(\z'_u - \z^t_u)^+ > 0$ only for
  nodes $u \not \in T_a$; since $\z'_u = \z^{t-1}_u$ for those nodes, we
  immediately get
  % so in $\|z' - z^t\|_{q}^+$ it
  % only matters the $q_u$'s for $u \notin T_a$. But we claim that for
  % such $u$'s $q(z',z^t)_u = q(z^{t-1},z^t)_u$:
  %       	%
  \begin{align*}
    q(\z',\z^t)_u = \alpha(\z'_{p(u)},\z^t_{p(u)}) - \alpha(\z'_u,\z^t_u) =
    \alpha(\z^{\alert{t-1}}_{p(u)},z^t_{p(u)}) - \alpha(\z^{\alert{t-1}}_u,\z^t_u) =
    q(\z^{t-1},\z^t).
  \end{align*}
  % where ``$*$'' follows because $z'$ and $z^{t-1}$ are the same
  % outside of $T_a$ (and notice that if $u \notin T_a$, then
  % $p(u) \notin T_a$).
		
  Finally, using the fact that we send flows from other leaves to the
  requested node $r_t$, the solutions satisfy $\z^t_u \le \z'_u \le
  \z^{t-1}_u$ for all nodes outside the backbone,
   we have
  \begin{align*}
  	\|\z' - \z^t\|_{q(\z^{t-1},\z^t)}^+ &= \sum_{u \notin \textrm{backbone}} w_u q(\z^{t-1},\z^t)_u \cdot (\z'_u - \z^t_u)^+ \\
  	&\le \sum_{u \notin \textrm{backbone}} w_u q(\z^{t-1},\z^t)_u (\z^{\alert{t-1}}_u - \z^t_u)^+ = \|\z^{t-1} - \z^t\|_{q(\z^{t-1},\z^t)}^+ .
 	\end{align*}
 	Putting these together concludes the proof.
\end{proof}

\subsubsection{Proof of Lemma~\ref{lem:main-tech}}

In this section we give the proof of Lemma~\ref{lem:main-tech}. Given a path $P$ from aמy leaf in the 
tree to $r_t$, we denote by $\top(P)$ its topmost vertex, which is always on the backbone. The
proof takes the residual flow $\z' - \z^t$, decomposes it as ``small''
flows on paths from leaves outside $T_a$ to the request node $r_t$, 
discharging them iteratively. We order these paths $P$ so that their
vertex $\top(P)$ becomes higher on the backbone over time, in order to
control the change in value of the $q$'s. (Observe that by pushing $\eps$
flow
over a path $P = u \leadsto \top(P) \leadsto r_t$ for a leaf
$u \neq r_t$ causes the $\z$-value of all but the last node in the
subpath $u \leadsto \top(P)$ to drop by $\eps$, and the $\z$-value of all
but the first node in the subpath $\top(P) \leadsto r_t$ to increase by
$\eps$.)

Since several flow paths may have the same $\top(\cdot)$ node, and the
$q$ function changes with each flow discharge, we handle this carefully. Consider a
sequence of flow paths that all have $\top(\cdot) = \ell$ that
transform $\bar{\z} \to \ldots \to \bar{\z}' \to \bar{\z}'' \to \ldots \to
\bar{\z}'''$, and we are presently concerned with a particular flow path
$P$ that takes us from $\bar{\z}'$ to $\bar{\z}''$. We say that an edge $(p(u),u)$ is \textbf{heavy for $\z$} if $\z_u > (2/3) \z_{p(u)}$, i.e., it contributes to the potential $\apot$. (If we have equality $\z_u = (2/3) \z_{p(u)}$, we
  may or may not consider this edge heavy).

\begin{lemma}[One Path]
  \label{lem:onestep}
  Consider LP solutions $\bar{\z},\bar{\z}',\bar{\z}'', \bar{\z}'''$ and a
  backbone node $\ell$ (with a child $v'$ on the backbone, and some child
  $v$ outside the backbone) such that:
  \begin{OneLiners}
  \item[(i)] $\bar{\z}_\ell \ge \frac{1}{10}$,
  \item[(ii)] $\bar{\z}_\ell = \bar{\z}'_\ell = \bar{\z}''_\ell = \bar{\z}'''_\ell$ and
    $\bar{\z}'''_v \le \bar{\z}''_v \le \bar{\z}'_v \le \bar{\z}_v$,
    % is obtained from $\bar{z}$ by \alert{pushing flow along multiple
    % paths with $\top(.) = \ell$}
  \item[(iii)] $\bar{\z}''$ is obtained from $\bar{\z}'$ by pushing $\eps$
    flow over a leaf-to-$r_t$ path $P$ that has $\top(P) = \ell$ and that passes through $v$ and $v'$, and
  \item[(iv)] the heavy edges in $\bar{\z}'$ and $\bar{\z}''$ are the
    same, where edges with $\z_u = (2/3) \z_{p(u)}$ may be considered
    heavy or not, as needed.
  \end{OneLiners}
  Then
  \begin{gather}
    \|\bar{\z}' - \bar{\z}''\|_w^+ \le c \cdot \|\bar{\z}' -
    \bar{\z}''\|^+_{q(\alert{\bar{\z}},\alert{\bar{\z}'''})} -
    2\left[\apot(\bar{\z}'') - \apot(\bar{\z}')\right]. \label{eq:2}
  \end{gather}
\end{lemma}
	
\begin{proof}
  Observe that the left hand side of~(\ref{eq:2}) is at most
  $w_{v} \eps (1 + \tau + \tau^2 + \ldots) = \frac{w_v \e}{1-\tau}$. To
  bound the right hand side from below, we consider two cases, based on which of
  the two topmost edges in $P$ may be heavy.
		
  \emph{Case 1: $\bar{\z}'_v \ge \frac{2}{3} \bar{\z}'_\ell$.} Hence the
  edge $(\ell,v)$ is heavy. Since the first term on the right hand side is
  non-negative, it is at least
  $-2 \left[\apot(\bar{\z}'') - \apot(\bar{\z}') \right] = 2 \sum_u
  -\left[\apot(\bar{\z}'')_u - \apot(\bar{\z}')_u \right]$, where
  $\apot(y)_u = w_u (y_u - (2/3) y_{p(u)})^+$. Let $P_+$ and $P_-$ be the
  subpaths of $P$ from $r_t$ to $\top(P)$, and down from there, such
  that $z$ increases on the former and decreases on the latter, as we
  push flow.

  Now $-\left[\apot(\bar{z}'')_v - \apot(\bar{z}')_v \right] = w_v \e$;
  this potential change for vertices $u$ on $P_+$ (apart from vertex
  $v'$) is $-w_u \e/3$, and this change for vertices whose parents lie
  on $P_-$ is $-2w_u \e/3$. Since all these vertices $u$ are at a
  greater depth than $v$ is, using the HST property means that the total
  change in potential is at least
  $2w_v \e(1 - \tau - \tau^2 - \ldots) = 2w_v\e
  \frac{1-2\tau}{1-\tau}$. Since $\tau \leq 1/10$, the right hand side is at least
  $\frac{w_v \e}{1-\tau}$ and hence the left hand side in this case.

  \medskip \emph{Case 2: $\bar{\z}'_v \le \frac{2}{3} \bar{\z}'_\ell$.}
  In this case the edge $(\ell,v)$ is not heavy, but $(\ell,v')$ may be
  heavy. Hence the expression
  $-\left[\apot(\bar{\z}'')_u - \apot(\bar{\z}')_u \right]$ is at least
  $-w_v \e$ for $v'$. Again the same calculations as in the previous
  case imply
  $- 2\left[\apot(\bar{\z}'') - \apot(\bar{\z}') \right] \geq - 2\frac{w_v
    \eps}{1-\tau}$. Hence, for the right hand side to be larger than the left hand side, we
  need to take the first term on the right hand side into account. This is
  $c \cdot \|\bar{\z}' - \bar{\z}''\|^+_{q(\bar{\z},\bar{\z}''')} \ge c
  \cdot w_v \cdot q(\bar{\z},\bar{\z}''')_v \cdot (\bar{\z}'_v -
  \bar{\z}''_v)^+ = c \cdot w_v \cdot q(\bar{\z},\bar{\z}''')_v \cdot
  \eps$. For this to be at least $3\frac{w_v \eps}{1-\tau}$, it suffices
  to show $q(\bar{\z},\bar{\z}''')_v$ is at least a constant, since we can
  then set $c$ large enough.
		
	To simplify the notation, let $f(x) = \ln(1+ \nicefrac{x}{\delta})$. Using property~(ii), $\bar{z}_\ell = \bar{z}'''_\ell$ and so we have
  \begin{align}
    q(\bar{\z},\bar{\z}'')_v = \alpha(\bar{\z}_\ell,\bar{\z}'''_\ell) -
    \alpha(\bar{\z}_v,\bar{\z}'''_v) = f(\bar{\z}_\ell) -
    \alpha(\bar{\z}_v,\bar{\z}'''_v),   \label{eq:q}
  \end{align}
  and we now want to upper-bound the last term. Using again property~(ii),
  $\bar{\z}'''_v \le \bar{\z}''_v \le \bar{\z}'_v \le \bar{\z}_v$, so
  $$\alpha(\bar{\z}_v,\bar{\z}'''_v) = \alpha(\bar{\z}'''_v,\bar{\z}_v) \le
  \alpha(\bar{\z}'_v, \bar{\z}_v).$$ Moreover, in the current case~2,
  $\bar{\z}'_v \le \frac{2}{3} \bar{\z}'_\ell = \frac{2}{3} \bar{z}_\ell$
  and $\bar{\z}_v \le \bar{\z}_\ell$,
  so $\alpha(\bar{\z}'_v, \bar{\z}_v) \le \alpha(\frac{2}{3}\bar{\z}_\ell ,
  \bar{\z}_\ell)$. Now letting $X$ be a random variable uniformly
  distributed in $[\frac{2}{3} \bar{\z}_\ell, \bar{\z}_\ell]$ (with median
  $med(X) = \frac{5}{6} \bar{\z}_\ell$) we have
  \begin{align*}
    \alpha((2/3)\bar{\z}_\ell , \bar{\z}_\ell) = \E f(X) &= \nicefrac{1}{2} \E[f(X) \mid X \le med(X)] + \nicefrac{1}{2} \E[f(X) \mid X \ge med(X)]\\
                                                  &\le \nicefrac{1}{2} f(med(X)) + \nicefrac{1}{2} f(\bar{\z}_\ell)
                                                  = \nicefrac{1}{2} f((\nicefrac56) \bar{\z}_\ell) + \nicefrac{1}{2} f(\bar{\z}_\ell).
  \end{align*}
  Putting everything together and applying it to \eqref{eq:q} we get
  \begin{align*}
    q(\bar{\z},\bar{\z}''')_v \ge f(\bar{z}_\ell) - \left(\nicefrac{1}{2}
    f((5/6) \bar{\z}_\ell) + \nicefrac{1}{2} f(\bar{\z}_\ell)\right) =
    \nicefrac{1}{2} \bigg(f(\bar{\z}_\ell) - f((5/6) \bar{\z}_\ell)\bigg) =
    \nicefrac{1}{2} \bigg(\ln \left(\frac{1 + \bar{\z}_\ell/\delta}{1 + (5/6)
    \bar{\z}_\ell /\delta}\right)\bigg).
  \end{align*}
  This final expression is at least a constant, because $\bar{\z}_\ell \ge
  \nicefrac{1}{10}$ by property~(i). This concludes the proof.
\end{proof}

Using Lemma~\ref{lem:onestep} repeatedly, we can understand the process
of going from  $\bar{\z} \to \ldots \to \bar{\z}'$ via pushing flow
on multiple paths, as long as all these paths share the same top node.
\begin{lemma}[All Paths with $\top(P) = \ell$]
  \label{lem:onetop}
  Consider LP solutions $\bar{\z},\bar{\z}'$ and a node $\ell$ on the
  backbone such that: (i)~$\bar{\z}_\ell \ge \frac{1}{10}$,
  and~(ii)~$\bar{\z}'$ is obtained from $\bar{\z}$ by pushing flow along
  multiple paths with $\top(.) = \ell$.
  Then
  $$\|\bar{\z} - \bar{\z}'\|_w^+ \le c \cdot \|\bar{\z} -
  \bar{\z}'\|^+_{q(\bar{\z},\bar{\z}')} - 2\left[\apot(\bar{\z}') -
    \apot(\bar{\z})\right].$$
\end{lemma}
	
\begin{proof}
  Consider a sequence of solutions
  $\bar{\z} = \z_0, \z_1,\ldots, \z_m = \bar{\z}'$, where $\z_{i+1}$ is
  obtained from $\z_i$ by pushing flow over a leaf-to-$r_t$ path $P_i$ with
  $\top(P_i) = \ell$; we construct this sequence of flows such that the
  set of heavy edges with respect to each consecutive pair $\z_i$ and
  $\z_{i+1}$ is the same, for all $i$.
		
  Applying Lemma~\ref{lem:onestep} with $\bar{\z} = \z_0$,
  $\bar{\z}' = \z_i$, $\bar{\z}'' = \z_{i+1}$, and $\bar{\z}''' = \z_m$, we
  get
  \begin{align*}
    \|\z_i - \z_{i+1}\|_w^+ \le c \cdot \|\z_i - \z_{i+1}\|^+_{q(\z_0,\z_m)} -
    2\left[\apot(\z_{i+1}) - \apot(\z_i)\right].
  \end{align*}
  Summing this up over all indices $i$, and observing that
  $\sum_i \|\z_i - \z_{i+1}\|^+_w = \|\z_0 - \z_m\|^+_w$ (since for each
  vertex $u$, every flow either increases its value or decreases it),
  concludes the proof.
\end{proof}
	
Now we can finally give the proof of Lemma~\ref{lem:main-tech}.	
\begin{proof}[Proof of Lemma~\ref{lem:main-tech}]
  Recall the node $a$ was defined to be the highest node on the backbone
  with $\z^{t-1}$-value at most $1/10$. Let $a, v_1, v_2, \ldots, v_m$ be
  the vertices on the backbone starting from $a$, going from bottom to top. In
  order to change the solution $\z'$ into $\z^t$, push the remaining flow
  on paths ``from bottom to top'', i.e., those with $\top(\cdot) = v_1$
  first, then those with $\top(\cdot) = v_2$ next, etc., to get a
  sequence of solutions
  $\z' = \z_0 \rightarrow \z_1 \rightarrow \ldots \rightarrow \z_m = \z^t$.
	
  Since all these remaining flows have $\top(.)$ higher than node $a$, we have
  $z_{\top(.)} \ge \frac{1}{10}$ at all times during this process, so
  Lemma~\ref{lem:onetop} applied to each choice of $\top(\cdot)$ gives
  \begin{align}
    \sum_i \|\z_i- \z_{i+1}\|_w^+ \le c \cdot \sum_i \|\z_i -\z_{i+1}\|_{q(\z_i,\z_{i+1})}^+ - 2 \left[\apot(\z_m) - \apot(\z_0) \right]. \label{eq:last}
  \end{align}
  Again each intermediate step causes the $\z$-value on the backbone to   rise, and off the backbone to fall: this monotonicity means the left
hand side
  equals $\|\z_0 - \z_m\|_w^+$.
		
  Now we claim that
  $$\|\z_i - \z_{i+1}\|_{q(\z_i,\z_{i+1})}^+ \le \|\z_i
  -\z_{i+1}\|_{q(\alert{\z_0,\z_m})}^+.$$Since $(\z_i)_u \ge (\z_{i+1})_u$
  only in the subforest $T_{v_i} \setminus (\{v_i\} \cup T_{v_{i-1}})$,
  it suffices to show that $q(\z_i,\z_{i+1})_u \le q(\z_0,\z_m)_u$ for nodes
  $u$ in this set. In fact, the flows that convert $\z_i$ to $\z_{i+1}$
  are the only ones that change the $z$-value of these nodes, so we have
  $(\z_0)_u = (\z_i)_u$ and $(\z_{i+1})_u = (\z_m)_u$. Now since the
  $q$-value at node $u$ depends on the $\z$-values at the node and its
  parent, we have equality $q(\z_i,\z_{i+1})_u = q(\z_0,\z_m)_u$ for all
  nodes in $T_{v_i} \setminus (\{v_i\} \cup T_{v_{i-1}})$, except perhaps when $u$ is a child of $v_i$. For
  $u$ being a child of $v_i$ (and not on the backbone), we have
  \begin{align*}
    q(\z_i,\z_{i+1})_u &= \alpha((\z_i)_{v_i}, (\z_{i+1})_{v_i}) - \alpha((\z_i)_u,
                       (\z_{i+1})_u) \\ &= \alpha((\z_i)_{v_i}, (\z_{i+1})_{v_i}) -
                       \alpha((\z_{\alert{0}})_u, (\z_{\alert{m}})_u)\\
    &= \alpha((\z_{\alert{0}})_{v_i},
  (\z_{i+1})_{v_i}) - \alpha((\z_{0})_u, (\z_{m})_u) \tag{since $(\z_0)_{v_i} = (\z_i)_{v_i}$} \\
    &\leq \alpha((\z_{0})_{v_i},
  (\z_{\alert{m}})_{v_i}) - \alpha((\z_{0})_u, (\z_{m})_u) = q(\z_0,\z_m)_u ,
  \end{align*}
  where the inequality uses the fact that
  $(\z_m)_{v_i} \ge (\z_{i+1})_{v_i} \ge (\z_0)_{v_i}$.  This proves the claim.

  Thus, the right hand side of \eqref{eq:last} is at most
  $$c \cdot \sum_i \|\z_i -\z_{i+1}\|_{q(\z_0,\z_m)}^+ - 2 \left[\apot(\z_m) -
    \apot(\z_0) \right].$$ Again, since there is no cancellation in the
  flows, the sum equals $\|\z_0 -\z_m\|_{q(\z_0,\z_m)}^+$. This concludes
  the proof.
\end{proof}

%%% Local Variables:
%%% mode: latex
%%% TeX-master: "main"
%%% End:

\subsection*{Acknowledgments}

We thank Nikhil Bansal, S\'ebastien Bubeck, Amit Kumar, James Lee,
Aleksander Madry, and Harry R\"{a}cke for helpful discussions. In
particular we thank S\'ebastien, James, and Aleksander for generously
explaining their work % at the Simons Institute for Theoretical Computer
% Science
to audiences containing different subsets of the authors, for clarifying
issues related to discretization of their algorithm, and for 
pointing out a gap in Lemma~\ref{lem:opt-charge}. Part of this work was done when the authors were
visiting the \emph{Algorithms and Uncertainty}
and \emph{Bridging Discrete and Continuous Optimization} programs at the
Simons Institute for the Theory of Computing.

{\small
\bibliography{bib}
\bibliographystyle{amsalpha}
}

\newpage
\appendix

\section{The Weighted $(h,k)$-Paging Problem}
\label{sec:paging}

In this section we consider the weighted $(h,k)$-paging problem, where we maintain a
fractional solution with $k$ pages, and compare to an optimal solution
that maintains $h \leq k$ pages.  The goal is to respond to requests
$r_t \in [n]$ at each time $t$ by producing a vector $\z^t \in \{0,1\}^n$
with $\norm{\z^t}_1 = k$, such that $\z^t_{r_t} = 1$ for all $t$. Define
the \emph{weighted $\ell_1$ norm} for $x \in \R^n$ by
\[ \wlone{x} := \sum_i w_i |x_i|. \] The
objective is to minimize the total weighted movement cost:
\[ \sum_t \wlone{\z^t - \z^{t-1}}. \]
This problem is equivalent to the $(h,k)$-server problem on a weighted
star metric.

\subsection{Polytopes and Solutions}

As is common for fractional paging, consider the \emph{anti-paging} polytope
\[ P := \{ x \in [0, 1]^n \mid \sum_i x_i \geq n-h \}. \] For
$\delta := \frac{k-h+\nicefrac12}{k+\nicefrac12}$, define the \emph{shifted} polytope
\[ P_\delta := \{ x \in [\delta, 1]^n \mid \sum_i x_i \geq n-h \}. \]
We maintain the following invariant:

\begin{invariant}
  \label{inv:paging}
  The algorithm's solutions are fractional vectors $\a^t \in P_\delta$
  with $\|\a^t\|_1 = n-h$, and the optimal solutions are Boolean vectors
  $\b^t \in P \cap \{0,1\}^n$, again with $\|\b^t\|_1 = n-h$. Moreover,
  for $t \geq 1$, $\a^t_{r_t} = \delta$ and $\b^t_{r_t} = 0$.
\end{invariant}

At the beginning, if the optimal servers are at some set $B^0 \sse [n]$
with $|B_0| = h$, define $\b^0_i = \mathbf{1}_{i \not\in B^0}$. If the
algorithm's servers are initially at $A^0$ with $|A^0| = k$, define
$\a^0_i = \delta$ for all $i \in A^0$, and
$\a^0_i = \frac{n-h - \delta k}{n-k}$ for $i \not\in A^0$. It is easy to
verify that $\a^0, \b^0$ satisfy the invariant.

Interpreting the Boolean vector $\b^t$ is simple: the adversary has
servers exactly at the $h$ locations $i \in [n]$ where $\b^t_i = 0$,
i.e., its paging solution is given by $\mathbf{1} - \b^t$. On the other
hand, converting the algorithm's fractional solution $\a^t \in P_\delta$
from the shifted anti-paging polytope to a fractional paging solution
$\z^t$ requires handling this shift: define
$\z^t_i := \frac{1-\a^t_i}{1-\delta}$ for each $i \in [n]$. This new
fractional solution $\z^t$ has $\z^t_{r_t} = 1$ since
$\a^t_{r_t} = \delta$, and it uses
\[ \|\z^t\|_1 = \frac{n-\norm{\a^t}_1}{1-\delta} = \frac{n-(n-h)}{h/(k+\nicefrac12)} = k+\nicefrac12 \]
servers. It satisfies $\wlone{\z^t - \z^{t-1}} = \wlone{\a^t - \a^{t-1}}
\cdot \nicefrac{1}{(1-\delta)}$.

We finally use the following result that combines~\cite[Lemma~3.4]{BCLLM17} with
\cite[\S5.2]{BBMN11} to round fractional solutions $\z^t$ to integer ones:
\begin{theorem}[Rounding Theorem]
  \label{thm:page-round}
  There exists an absolute constant $C > 1$ and an efficient randomized
  algorithm that takes a sequence of fractional solutions
  $\z^t = \frac{\mathbf{1} - \a^t}{1-\delta}$ to the weighted paging problem, each with
  $\norm{\z^t}_1 = k+\nicefrac12$ pages and with $\z^t_{r_t} = 1$, and
  rounds them to integer solutions $\hat{\z}^t$ each with
  $\norm{\hat\z^t} = k$ pages and $\hat\z^t_{r_t} = 1$, so that the
  expected movement cost is
  \[ \E\Big[ \wlone{\hat\z^t - \hat\z^{t-1}} \Big] \leq C\, \wlone{\z^t -
      \z^{t-1}} \leq \frac{C}{1-\delta}\, \wlone{\a^t - \a^{t-1}}. \]
\end{theorem}

Henceforth, we only consider the problem of maintaining the fractional
solution $\a^t \in P_\delta$. Our main theorem for computing
fractional solutions for weighted paging is the
following:
\begin{theorem}[Main Theorem: Weighted $(h,k)$-Paging]
  \label{thm:page-cr}
  There is an algorithm that maintains a sequence of fractional
  solutions $\a^t \in P_\delta$ with $\a^t_{r_t} = \delta$ and
  $\norm{\a^t}_1 = n-h$, such that
  \[ \wlone{\a^t - \a^{t-1}} \leq (\log 1/\delta) \cdot \wlone{\b^t -
      \b^{t-1}} + C', \] for any sequence of feasible solutions $\b^t$
  where $\b^t \in P \cap \{0,1\}^n$ and $\b^t_{r_t} = 0$. Here $C'$ is a
  constant that depends on the weights $w_i$ and the values of $k$ and
  $h$, but is independent of the request sequence.
\end{theorem}
Combining Theorems~\ref{thm:page-round} and~\ref{thm:page-cr}, and using
our choice of $\delta$, the competitive factor of our randomized
algorithm is $O(\log \frac{k+\half}{k - h + \half})$ for
$h = \Omega(k)$.  Note that for $h=k$ the algorithm is $O(\log k)$-competitive, and for $h=k/2$ the algorithm is $O(1)$-competitive.
In the rest of this section, we prove
Theorem~\ref{thm:page-cr}.

\subsection{The Projection Algorithm}

Given a request at location $r_t$, define the body
$P_t = P \cap \{x_{r_t} \leq \delta\}$.  Project the old point
$\a^{t-1} \in P_\delta$ onto the new body $P_t$ using a weighted form of the
\emph{unnormalized KL divergence}:
\[ \Ddiv{x}{x'} := \sum_i w_i \bigg(x_i \log \frac{x_i}{x'_i} -
  x_i + x'_i\bigg) \] I.e., set
$\a^t := \arg\min_{x \in P_t} \Ddiv{x}{\a^{t-1}}$. That's the entire algorithm.

Since $P_t$ is not contained within $P_\delta$, we must prove that the
new point $\a^t$  lies in the polytope $P_\delta$.

\begin{lemma}
  \label{lem:page-props}
  For each $t$, the solution $\a^t$ satisfies Invariant~\ref{inv:paging}.
\end{lemma}

\begin{proof}
  We assume that $\a^{t-1}$ satisfies the invariant, and then prove it
  for $\a^t$. As the base case, the invariant holds for $\a^0$.  Denote
  $\a := \a^{t-1}$ and $\a' := \a^t$, to avoid visual clutter. The
  projection operation that defines $\a'$ can be written as follows:
  \begin{align*}
    \min \quad &\sum_{i} w_i \Big(x_i \log \frac{x_i}{\a_i} -
                 x_i + \a_i\Big) \\
    \sum_i x_i &\geq n-h \\
    x_i &\leq 1 \qquad \forall i \neq r_t\\
    x_{r_t} &\leq \delta
  \end{align*}
  The KKT optimality conditions give us dual multipliers
  $\lambda, \gamma_i \geq 0$ that for all $i$,
  \begin{gather}
    w_i \log \frac{\a'_i}{\a_i} = \lambda - \gamma_i. \label{paging:kkt}
  \end{gather}
  Moreover, by complementary slackness, $\lambda > 0$ implies
  $\sum_i \a'_i = n-h$, and $\gamma_i >0$ implies $\a'_i = 1$ for
  $i \neq r_t$, and that $\a'_i = \delta$ for $i = r_t$.

  If $\a$ already belongs to $P_t$ then the projection returns
  $\a' = \a$, in which case the proof trivially follows. Hence assume
  that $\a_{r_t} > \delta$.  It now follows that:
  \begin{OneLiners}
  \item[(a)] $\a'_{r_t} = \delta$. Indeed, observe that
    $\a'_{r_t} \leq \delta < \a_{r_t}$, and by~(\ref{paging:kkt}) this
    decrease can only come about if $\gamma_{r_t} > 0$. Now
    complementary slackness for $\gamma_{r_t}$ implies
    $\a'_{r_t} = \delta$.
  \item[(b)] All other coordinates rise, i.e., $\a'_i \geq \a_i$ for all
    $i \neq r_t$. Indeed, if $\gamma_i > 0$, we have
    $\a'_i = 1 \geq \a_i$, else $\a'_i = \a_i e^\lambda \geq \a_i$ since
    $\lambda \geq 0$. In particular, if we start off with
    $\a_i \geq \delta$, the final solution also satisfies
    $\a'_i \geq \delta$.
  \item[(c)] $\|\a'\|_1 = n-h$. Indeed, inductively assume
    $\| \a \|_1 = n-h$. Then either $\a_{r_t} = \delta$ already, and
    hence there is no change, so $\a' = \a$. Else setting
    $\a'_{r_t} \gets \delta$ means we must increase some coordinates, so
    $\lambda > 0$ and hence $\| \a'\|_1 = n-h$. \qedhere
  \end{OneLiners}
\end{proof}

Since $\a^t_{r_t} = \delta$ from Lemma~\ref{lem:page-props}, let us give
an equivalent view of our algorithm. Define the \emph{auxiliary Bregman
  divergence}
\begin{gather}
  \Ddiv{x}{x'}^{wo} := \sum_{i \neq r_t} w_i \Big(x_i \log
  \frac{x_i}{x'_i} - x_i + x'_i\Big). \label{eq:aux-breg}
\end{gather}
Our algorithm is equivalent to first setting $\a^t_{r_t} \gets \delta$,
and then projecting the rest of the coordinates of $\a^{t-1}$ onto $P_t$
by $\min_{x \in P_t} \Ddiv{x}{\a^{t-1}}^{wo}$.

\subsection{Bounding the Movement Cost}

\paragraph{The Potential.} We use a KL-divergence-type potential to
measure the ``distance'' from the optimal solution $y^t \in P$ to $x^t$:
\begin{gather}
  \Pdiv{\b^t}{\a^t} := \sum_i w_i \, \b^t_i\, \log \frac{\b^t_i}{\a^t_i}
  = \sum_{i: \b^t_i = 1} w_i \log \frac{1}{\a^t_i}. \label{eq:page-pot}
\end{gather}
Observe that $\Phi \geq 0$ as long as $\y \in \{0,1\}^n$ and $\x \in (0,1]^n$.
Again, $\Ddiv{\b^t}{\a^t}$ could be used as the potential, but removing
the linear terms makes the arguments cleaner. Recall that this potential
has already been used in a potential function proof by Bansal et
al.~\cite{BBN-simple}; it arises naturally given our algorithm.

\paragraph{When OPT moves:}
Say OPT pays for moving pages out of the cache, i.e., for increasing
$\b$ from $0$ to $1$, it pays $w_i$. In this case the potential
increases by at most $w_i \log (1/\delta)$, because the denominator
$\a^{t-1}_i \geq \delta$. This gives us:
\begin{gather}
 \Delta \Phi_t^{OPT} := \Pdiv{\b^t}{\a^{t-1}} - \Pdiv{\b^{t-1}}{\a^{t-1}} \leq \log (1/\delta)
  \cdot \sum_i w_i (\b^t_i - \b^{t-1}_i)^+. \label{eq:page-opt-move}
\end{gather}

\paragraph{When ALG moves:} Recall the auxiliary Bregman
divergence~\eqref{eq:aux-breg}. The equivalent view of the algorithm
(discussed above) shows
\[ \a^t \in \arg\min_{x \in P_t} \Ddiv{x}{\a^{t-1}}^{wo}. \] Hence
$\a^t$ is a projection of $\a^{t-1}$ onto $P_t$ with respect to a
Bregman divergence $\Ddiv{}{}^{wo}$. Since the optimal solution $\b^t$
lies in $P_t = P\cap\{x_{r_t} \leq \delta\}$, the reverse-Pythagorean
property gives
\begin{gather}
  \Ddiv{\b^t}{\a^{t-1}}^{wo} \geq \Ddiv{\b^t}{\a^t}^{wo} + \Ddiv{\a^t}{\a^{t-1}}^{wo} \notag
  \\
  \implies \underbrace{\Pdiv{\a^t}{\a^{t-1}}^{wo}}_{\text{shadow cost}} +
  \underbrace{(\Pdiv{\b^t}{\a^t} - \Pdiv{\b^t}{\a^{t-1}})}_{\Delta \Phi_t^{ALG}} \leq
  0. \label{eq:shady}
\end{gather}
Here, the \emph{``without-$r_t$'' potential} is defined much as you would
expect: $\Pdiv{x}{x'}^{wo} := \sum_{i \neq r_t} w_i x_i \log
\frac{x_i}{x'_i}$. The second line in~(\ref{eq:shady}) follows from the
first by using the definition~(\ref{eq:aux-breg}) and canceling linear
terms on both sides. Moreover,
$\Pdiv{\b^t}{\cdot}^{wo} = \Pdiv{\b^t}{\cdot}$ since $\b^t_{r_t} = 0$.

Since all coordinates except for $r_t$ increase, we get
\begin{gather}
  \sum_i w_i (\a^t_i - \a^{t-1}_i)^+ = \sum_{i \neq r_t} w_i ( \a^t_i -
  \a^{t-1}_i ) \stackrel{\text{(\ref{eq:pmp})}}{\leq} \sum_{i \neq r_t}
  w_i\; \a^t_i \log \frac{\a^t_i}{\a^{t-1}_i} =
  \Pdiv{\a^t}{\a^{t-1}}^{wo}. \label{eq:paging-final}
\end{gather}

\subsubsection{Wrapping Up}

Combining~(\ref{eq:page-opt-move}),~(\ref{eq:shady})
and~(\ref{eq:paging-final}) gives us
\[ \sum_i w_i(\a^t_i - \a^{t-1}_i)^+  + \Big(\Pdiv{\b^t}{\a^t} -
  \Pdiv{\b^{t-1}}{\a^{t-1}}\Big) \leq (\log 1/\delta)\cdot \sum_i w_i(\b^t_i -
  \b^{t-1}_i)^+.\]
Summing up over all times $t$, and using the property that the final potential is
non-negative, we have
\begin{gather}
  \sum \sum_i w_i(\a^t_i - \a^{t-1}_i)^+ \leq (\log 1/\delta)\cdot
  \sum_t \sum_i w_i(\b^t_i - \b^{t-1}_i)^+ + \Pdiv{\b^0}{\a^0} . \label{eq:paging-abs-diff}
\end{gather}
We would like to translate the cost in terms of the weighted $\ell_1$
metric. For this, observe that for any sequence of numbers
$p_0, p_1, \ldots, p_T \in [0,M]$, we have
\[ \sum_t (p_t - p_{t-1})^+ \leq \sum_t |p_t - p_{t-1}| \leq 2 \sum_t (p_t
  - p_{t-1})^+ + M. \]
Applying this to each term in the summations from~(\ref{eq:paging-abs-diff}), we get
\[ \sum_t \wlone{\a^t_i - \a^{t-1}_i} \leq 2 \log(1/\delta) \cdot \sum_t \wlone{\b^t_i -
    \b^{t-1}_i} + O\Big(\sum_i w_i\Big) + \Pdiv{\b^0}{\a^0} .\]
Observe that the last term is at most $$\sum_{i \in A^0\setminus B^0} w_i
\log (1/\delta) \leq \sum_{i \in [n]} w_i \log (1/\delta),$$ and hence
setting $C' = O(\sum_{i \in [n]} w_i \log (1/\delta))$ completes the proof of Theorem~\ref{thm:page-cr}. \qed

An aside: while the above proof for paging proceeded via
the $\Ddiv{}{}^{wo}$ divergence, it could have instead
followed the arguments in \S\ref{sec:when-alg-moves} using
Lemma~\ref{lem:second-bound}, which would give very similar results.

%%% Local Variables:
%%% mode: latex
%%% TeX-master: "main"
%%% End:

\section{Proofs from Section~\ref{sec:server}}

\subsection{Omitted Proofs from Section~\ref{sec:prop-proj}}

We now give proofs of properties we claimed in \S\ref{sec:prop-proj},
as well as some supporting claims.

\ksRoottight*

\begin{proof}
  Constraint~(\ref{eq:ks-root}) yields
  $\x^t_{r,j} \geq \mathbf{1}_{(j > h)}$, so it remains to prove this is
  an equality. The previous solution $\x^{t-1}$ satisfies the equality (by
  induction on $t$),
  and~(\ref{eq:kktroot-ks}) implies that for
  $\x^t_{\rootvtx, j} > \x^{t-1}_{\rootvtx,j}$ we must have
  $\lambda_{\rootvtx, j} > 0$. But then complementary
  slackness~(\ref{eq:ks-csroot}) implies
  $\x^t_{\rootvtx, j} = \mathbf{1}_{(j > h)}$.
\end{proof}

\ksBox*

\begin{proof}
  %\cite[Lemma~3.5]{BCLLM17}
  The lower bound is by induction on the levels of the tree. For the
  base case, the root satisfies
  $\x^t_{\rootvtx, j} \geq 0$. For $\x^t_{u,j}$ with $u$ at depth $d$,
  its parent $p(u)$ satisfies $\x^t_{p(u),1} \geq 0$ by the induction hypothesis. Now
  constraint~(\ref{eq:ks-subsets}) for $S = \{(u,j)\}$ completes the
  inductive step.

  For the upper bound, suppose it does not hold. Then, choose the highest node $u$ in
  the tree for which some $\x^t_{u,j} > 1$. Since $\x^{t-1}_{u,j} \leq 1$ (by induction), this coordinate of $\x^t$ has increased. Hence
  $a_{u,j} > 0$, which means by~(\ref{eq:kkt-ks}) there exists some
  $S \sse \chi_{p(u)}$ such that $(u,j) \in S$ for which
  $\lambda_S > 0$. Now~(\ref{eq:tight}) implies that
  $\x^t(S) = \sum_{\ell \leq |S|} \x^t_{p(u),\ell}$. Thus
  \[ \x^t(S - \{(u,j)\}) = \x^t(S) - \x^t_{u,j} < \sum_{\ell \leq |S|}
    \x^t_{p(u),\ell} - 1 \leq \sum_{\ell \leq |S - \{(u,j)\}|}
    \x^t_{p(u),\ell}. \] The strict inequality uses that $\x^t_{u,j} > 1$, and
  the last inequality uses the fact that $\x^t_{p(u),\ell} \leq 1$ by our
  choice of $u$. But this violates constraint~(\ref{eq:ks-subsets}) of the
  convex program, yielding a contradiction.

  Secondly, $\x^{t-1}_{r_t,1} \geq \delta$ (by induction $\x^{t-1} \in P_{\delta}$), so if
  $\x^t_{r_t, 1} < \delta$ then its value has
  decreased. Then by~(\ref{eq:kktprime-ks}) we must have
  $\gamma_t > 0$, which would imply $\x^t_{r_t, 1} = \delta$ by
  complementary slackness~(\ref{eq:ks-req-cs}), a contradiction.

  For the last claim, for all non-$r_t$ leaves, (\ref{eq:kkt-ks}) has no
  $b_{u,j}$ terms. Thus $\x^t_{u,j} \geq \x^{t-1}_{u,j} \geq
  \delta$.
\end{proof}

\begin{restatable}[Monotonicity]{lemma}{ksMono}
  \label{lem:mono}
  $\x^t_{u,j} \leq \x^t_{u,j+1}$.
\end{restatable}

%\ksMono*

\begin{proof}
  %\cite[Lemma~3.6]{BCLLM17}
  Again assume that $\x^{t-1}_{u,j} \leq \x^{t-1}_{u,j+1}$.  In~(\ref{eq:kkt-ks})
  observe that $b_{u,j} \geq b_{u,j+1}$ because it sums up over more
  non-negative terms; this means $\x^t_{u,j}$ cannot be larger than $\x^t_{u,j+1}$ because of
  these terms. So we focus on the $a_{u,j}$ terms in~(\ref{eq:kkt-ks}).

  Suppose $\x^t_{u,j} > \x^t_{u,j+1}$, then it must be that
  $a_{u,j} > a_{u,j+1}$. So there is some $\lambda_S > 0$ with
  $(u,j) \in S$, yet $(u,j+1) \not\in S$. Consider
  $S' := (S - \{(u,j)\}) \cup \{(u,j+1)\}$. By constraint~(\ref{eq:ks-subsets})
  for the set $S'$,
  \[ \sum_{(v,\ell) \in S'} \x^t_{v,\ell} \geq \sum_{i \leq |S'|}
    \x^t_{p(u),i}= \sum_{i \leq |S|} \x^t_{p(u),i}
    \stackrel{\text{(\ref{eq:tight})}}{=} \sum_{(v,\ell) \in S}
    \x^t_{v,\ell}, \] where the last equality holds because $\lambda_S > 0$. This
  implies $\x^t_{u,j+1} \geq \x^t_{u,j}$, hence a contradiction.
\end{proof}

We say a set $S$ is \emph{tight} if
$\sum_{j \leq |S|} \x^t_{p(S), j} = \x^t(S) := \sum_{(v,\ell) \in S} \x^t_{v,\ell}$. % A
% node $u$ is \emph{tight at size $s$} and is \emph{tight at set $S$} if
% there exists a tight set $S \sse \chi_u$ of size $|S| = s$.
Let
$\calC_u$ denote the collection of tight sets in $\chi_u$, and
$\calC := \cup_u \calC_u$ be all the tight sets.

\begin{restatable}[Uncrossing]{lemma}{ksUncrossing}
  \label{lem:cross}
  For $S_1, S_2 \sse \chi_u$,
  $S_1, S_2 \in \calC \implies S_1 \cup S_2 \in \calC$. I.e., the union
  of tight sets with a common parent gives a tight set.
\end{restatable}

\begin{proof}
  %\cite[Lemma~3.7]{BCLLM17}
  Let $M := \max_{(v,j) \in S_1 \cup S_2} \x^t_{v,j}$
  be the largest value in the union (say it is in $S_1$) and let
  $(v^*, j^*) \in S_1$ achieving this maximum value. Consider any
  $(v,j) \in S_2 \setminus S_1$: since the set $S_1$ is tight and
  $S' := (S_1 \cup \{(v, j)\})\setminus \{(v^*,j^*)\}$ satisfies the
  constraint~(\ref{eq:ks-subsets}), we infer that $\x^t_{v,j} = M$. From feasibility
  for the set $S_1 \setminus \{(v^*,j^*)\}$, we know that
  $\x^t_{u,|S_1|} \geq M$ so using monotonicity of Lemma~\ref{lem:mono},
  \[ \sum_{i \leq |S_1 \cup S_2|} \x^t_{u,i} \geq \x^t(S_1) +
    |S_2\setminus S_1| \cdot M = \x^t(S_1 \cup S_2).\] The converse direction
  (inequality) follows from feasibility, and hence $S_1 \cup S_2$ is also
  tight.
\end{proof}

\ksFlow*

\begin{proof}
  %\cite[Lemma~3.8]{BCLLM17}
  We assume these properties hold for $\x^{t-1}$, and show
 them for $\x^t$.  For each vertex $u$, let $S_u \sse \chi_u$ be the largest
  tight set with respect to $\x^t$ in $\calC_u$. By definition of tight sets, for each depth $d$,
  \begin{gather}
    \sum_{u \in V_d} \sum_{j \leq |S_u|} \x^t_{u,j} = \sum_{u \in V_d}
    \x^t(S_u). \label{eq:11a}
  \end{gather}
  We claim that
  \begin{gather}
    \sum_{u \in V_d} \sum_{j > |S_u|} \x^t_{u,j} \geq \sum_{(v,\ell) \in
      V_{d+1} \setminus (\cup_{u \in V_d} S_u)} \x^t_{v,\ell}. \label{eq:11b}
  \end{gather}
  By the constraints~(\ref{eq:ks-subsets}), we have
  $\sum_{u \in V_d} \sum_{j \leq |S_u|} \x^{t-1}_{u,j} \leq \sum_{u \in
    V_d} \x^{t-1}(S_u)$. Since we inductively assumed the
  $\x^{t-1}$-mass at each level of the tree was the same, collecting the
  terms not appearing in the above inequality gives us
  \begin{gather}
    \sum_{u \in V_d} \sum_{j > |S_u|} \x^{t-1}_{u,j} \geq \sum_{(v,\ell)
      \in V_{d+1} \setminus (\cup_{u \in V_d} S_u)} \x^{t-1}_{v,\ell}. \label{eq:11c}
  \end{gather}
  Now, for any $\x^t_{u,j}$ on the left hand side of~(\ref{eq:11b}), every set
  $T \sse \chi_u$ with $j \leq |T|$ is not tight, and has
  $\lambda_T = 0$. This means $b_{u,j} = 0$ in~(\ref{eq:kkt-ks}) and
  hence $\x^t_{u,j}$ is increasing, i.e.,
  $\x^t_{u,j} \geq \x^{t-1}_{u,j}$. Moreover, each term on the right
  side of~(\ref{eq:11b}) has $a_{v,\ell} = 0$ and is decreasing, i.e.,
  $\x^t_{u,j} \leq \x^{t-1}_{u,j}$: from the maximality of $S_u$ and Lemma \ref{lem:cross} all sets $T$ containing $(v,\ell)$ are non-tight and hence have $\lambda_T = 0$. Using these inequalities
  in~(\ref{eq:11c}) gives us~(\ref{eq:11b}). And together
  with~(\ref{eq:11a}) gives us
  $\sum_{u \in V_d} \sum_j \x^t_{u,j} \geq \sum_{u \in V_{d+1}} \sum_j
  \x^t_{u,j}$. From~(\ref{eq:ks-subsets}) we have the converse direction
  (inequalities)
  $\sum_j \x^t_{u,j} \leq \sum_{(v,\ell) \in \chi_u} \x^t_{v, \ell}$ for
  each node $u \in V_d$, so each such inequality must be tight.

  Since each leaf $u \neq r_t$ has $a_{u,1}$ terms but no $b_{u,1}$
  terms in~(\ref{eq:kkt-ks}), $\x^t_{u,1} \geq \x^{t-1}_{u,1}$. Since we
  just proved that the $\x^t$-value at each node equals the $\x^t$-value
  at its children, there is a ``flow'' of $\x^t$-measure from $r_t$ to
  all the other leaves. In particular, for each node not on the
  $\rootvtx$-$r_t$ path, $\sum_j \x^t_{u,j} \geq \sum_j \x^{t-1}_{u,j}$.
\end{proof}

\ksLevels*

\begin{proof}
  Expanding the first summation gives
  \begin{align*}
    \sum_{j} B^t_{u,j}
    = \sum_{j} \t\x^t_{u,j}  \sum_{T \sse \chi_u: |T| \geq j}
      \lambda_T = \sum_{T \sse \chi_u} \lambda_T \sum_{j \leq |T|}
      \t\x^t_{p(T), j}.
  \end{align*}
  By complementary slackness~(\ref{eq:tight}), we have
  $\lambda_T \sum_{j \leq |T|} \x^t_{p(T), j} = \lambda_T \sum_{(u,j)
    \in T} \x^t_{u, j} = \lambda_T \, \x^t(T)$. However, since there are
  $|T|$ terms on both sides, adding $\lambda_T |T|\d$ to both sides gives us the
  desired inequality for the shifted variables $\t\x$. Now summing up over
  all $T$ gives
  \begin{gather}
    \sum_{T \sse \chi_u} \lambda_T \sum_{j \leq |T|} \t\x^t_{p(T), j} =
    \sum_{T \sse \chi_u} \lambda_T \; \t\x^t(T). \label{eq:9}
  \end{gather}
  And the second summation in the statement of the lemma is
  \begin{gather}
    \sum_{(v,\ell) \in \chi_u} A^t_{v,\ell} = \sum_{(v,\ell)
      \in \chi_u} \t\x^t_{v,\ell} \sum_{S \sse \chi_u: (v,\ell) \in S}
    \lambda_S = \sum_{S \sse \chi_u} \lambda_S \sum_{(v,\ell) \in S}
    \t\x^t_{v,\ell}  = \sum_{S \sse \chi_u} \lambda_S \; \t\x^t(S).
    \label{eq:8}
  \end{gather}
  The two expressions are equal, hence the claim.
% For
%   the analogous statement about $\xt$, since every variable gets the
%   same shift $\delta$,~(\ref{eq:tight}) also holds for the $\xt$
%   variables.
\end{proof}

\subsection{Proof of Lemma~\ref{lem:second-bound}}
\label{sec:proof-secondbound}

Recall the statement of Lemma~\ref{lem:second-bound}, where $W(\x) :=
\sum_u w_u \sum_j \x_{u,j}$.
\lemSecond*

If $(\lambda, \gamma_t)$ are the optimal dual variables for the projection problem \eqref{eq:ks-root}-\eqref{eq:ks-req}, we can get an equivalent characterization of the  optimal solution
$\x^t$ as follows:

\begin{claim}
  \label{clm:equiv}
  Vector $\x^t$ is an optimal solution for the optimization problem
  \begin{gather}
    \min_{x \in P} \Ddiv{x}{\x^{t-1}}
  +\gamma_t x_{r_t}. \label{eq:12}
  \end{gather}
\end{claim}

\begin{proof}
  Recall that we defined $\x^t$ to be the minimizer of just the first term,
  subject to the constraints that $x \in P$ and $x \leq \delta$. If we now
  Lagrangify the second constraint, and use the fact that $\gamma_t$ is an
  optimal Lagrange multiplier, we get the equivalent problem
  $\min_{x \in P} \Ddiv{x}{\x^{t-1}} +\gamma_t x_{r_t}$. 
\end{proof}

Observe that $\Ddiv{x}{x'}$ is the Bregman divergence corresponding to
the strongly convex function
\[ h(x) := \sum_u w_u \sum_j \t{x}_{uj} \log \t{x}_{uj} . \]
Now, using the definition of Bregman divergences and some simple
algebra, we get that
\begin{align}
  \Ddiv{\y^t}{\x^{t-1}} - \Ddiv{\y^t}{\x^t}
  &= \Ddiv{\x^t}{\x^{t-1}} + \ip{ \gr
    h(\x^{t-1}) - \gr h(\x^t), \x^t  - \y^t } \notag\\
  &\geq \ip{ \gr
    h(\x^{t-1}) - \gr h(\x^t), \x^t - \y^t },  \label{eq:14}
\end{align}
where the inequality uses non-negativity of Bregman divergences.
Manipulating the linear terms, 
\begin{align}
  \Pdiv{\y^t}{\x^{t-1}} - \Pdiv{\y^t}{\x^t}
  &\geq \ip{ \gr
    h(\x^{t-1}) - \gr h(\x^t), \x^t - \y^t } + W(\x^t) - W(\x^{t-1}).  \label{eq:14a}
\end{align}
To prove Lemma~\ref{lem:second-bound}, we need to bound the inner
product term from below. For a point $u \in P$, define the normal cone at $u$ to be
$N_{P}(u) := \{d \mid \ip{d,v-u} \le 0 \;\forall v \in P\}$.

\begin{claim}
  \label{clm:normals}
  Let $\mathbf{e}_{r_t} \in \{0,1\}^N$ be the vector that has a $1$ in
  the coordinate corresponding to leaf $r_t$, and 0s otherwise. Then
  \[ \gr h(\x^{t-1}) - \gr h(\x^t) = \gamma_t \mathbf{e}_{r_t} + d \]
  where $d$ belongs to the normal cone $N_{P}(\x^t)$.
\end{claim}
\begin{proof}
  Since $\x^t$ solves the optimization problem~(\ref{eq:12}), the
  first-order optimality criteria implies that the gradient of the
  objective function, when evaluated at $\x^t$, belongs to the
  negative normal cone $- N_{P}(\x^t)$. Recall that the gradient at any
   point $x$ is
  \[ \gr\Ddiv{x}{\x^{t-1}} + \gamma_t \mathbf{e}_{r_t} = \gr h(x)
    - \gr h(\x^{t-1}) + \gamma_t \mathbf{e}_{r_t}, \] so
  $\gr h(\x^t) - \gr h(\x^{t-1}) + \gamma_t \mathbf{e}_{r_t} = -d$,
  where $d \in N_{P}(\x^t)$. Rearranging completes the proof.
\end{proof}

Substituting this expression into~(\ref{eq:14a}) implies that the inner
product term is
\begin{gather}
  \ip{ \gamma_t \mathbf{e}_{r_t} + d , \x^t - \y^t } \geq \ip {\gamma_t
    \mathbf{e}_{r_t}, \x^t - \y^t } = \gamma_t (\x^t_{r_t} - \y^t_{r_t}) =
 \gamma_t \delta.
\end{gather}
The inequality uses the definition of the normal cone $N_P(\x^t)$ and
that $\y^t \in P$.  The last equality uses  $\x^t_{r_t} = \delta$
and $\y^t_{r_t} = 0$. This proves Lemma~\ref{lem:second-bound}.

\subsection{Miscellaneous Lemmas}
\label{sec:misc-lemmas}

\begin{lemma}
  \label{lem:emd-to-lone}
  Consider vectors $\x,\x' \in P$, and let $\x,\x' \in \R^n$ be their respective restrictions to the leaf atoms. Then \[ d(\x, \x') \leq \,\Tlone{\x - \x'}. \] Moreover, if $\x,\x'$ are integer vectors, we get equality above.
%M: Changed the statement, I think we need at least flow conservation
%  For vectors $\x, \x' \in \R^N$, let $\x, \x' \in \R^n$ be the
%  restrictions of these vectors to just the leaf atoms. Then
%  \[ d(\x, \x') \leq \,\Tlone{\x - \x'}. \]
%  Moreover, if $\x, \x'$ are integer vectors and belong to $P$, we get
%  equality above.
\end{lemma}

\begin{proof}
	From the flow conservation Lemma \ref{lem:flow}, for each node $u$ we have $\sum_j \x_{u,j} = \sum_{v \in leaves(T_u)} \x_{v,1} = \sum_{v \in leaves(T_u)} \x_{v,1}$, and the same holds for $\x'$ and $\x'$. Therefore
	\begin{align}
		d(\x,\x') &= \sum_u w_u\, \Bigg|\sum_{v \in leaves(T_u)} \x_{v,1} - \sum_{v \in leaves(T_u)} \x'_{v,1}\Bigg| \notag\\
			&=\sum_u w_u\, \Bigg|\sum_j (\x_{u,j} - \x'_{u,j})\Bigg| \notag\\
			&\le \sum_u w_u \sum_j |\x_{u,j} - \x'_{u,j}| = \Tlone{\x - \x'}, \label{eq:emdLone}
	\end{align}
	concluding the first part of the proof.
	
	For the second part, when $\x$ and $\x'$ are integral $|\sum_j (\x_{u,j} - \x'_{u,j})|$ equals $|\textrm{\#1's in $(\x_{u,j})_j$} - \textrm{\#1's in $(\x'_{u,j})_j$}|$. Moreover, by the monotonicity Lemma \ref{lem:mono}, $\sum_j |\x_{u,j} - \x'_{u,j}|$ equal the same quantity. Thus, inequality \eqref{eq:emdLone} holds at equality and hence $d(\x,\x') = \Tlone{\x-\x'}$. This concludes the proof.
\end{proof}

%%% Local Variables:
%%% mode: latex
%%% TeX-master: "main"
%%% End:

\end{document}